%% file: manuscript.tex
\documentclass[12pt,english]{article}
\input{preamble}

\theoremstyle{remark}
\newtheorem{remark}{Remark}
\newtheorem{assumption}{Assumption}
\theoremstyle{definition}
\newtheorem{definition}{Definition}
\theoremstyle{lemma}
\newtheorem{lemma}{Lemma}

\usepackage{booktabs}
\usepackage{amsfonts}

\title{Fostering Data Collaboration in Digital Transportation Marketplaces: The Role of Privacy-Preserving Mechanisms}

\def\shortauthors{Wang et al.}
\def\runningtitle{Fostering Data Collaboration in Digital Transportation Marketplaces}

\author[a]{Qiqing Wang}
\author[b]{Haokun Yu}
\author[a$\ast$]{Kaidi Yang}
\affil[a]{Department of Civil and Environmental Engineering, National University of Singapore, 1 Engineering Drive 2, Singapore 117576, Singapore}
\affil[b]{Institute of Operations Research and Analytics, National University of Singapore, Innovation 4.0 Research Link 3, Singapore 117602, Singapore}

\def\corresemail{kaidi.yang@nus.edu.sg}

\def\abstract{ 
Data collaboration between municipal authorities (MA) and mobility providers (MPs) has brought tremendous benefits to transportation systems in the era of big data. 
Engaging in collaboration can improve the service operations (e.g., reduced delay) of these data owners, however, it can also raise privacy concerns and discourage data-sharing willingness. 
Specifically, data owners may be concerned that the shared data may leak sensitive information about their customers' mobility patterns or business secrets, resulting in the failure of collaboration.
This paper investigates how privacy-preserving mechanisms can foster data collaboration in such settings. We propose a game-theoretic framework to investigate data-sharing among transportation stakeholders, especially considering perturbation-based privacy-preserving mechanisms. 
Numerical studies demonstrate that lower data quality expectations can incentivize voluntary data sharing, improving transport-related welfare for both MAs and MPs.
Our findings provide actionable insights for policymakers and system designers on how privacy-preserving technologies can help bridge data silos and promote collaborative, privacy-aware transportation systems.
}

\def\keyword{Digital transportation marketplace \\ Data collaboration game\\Data privacy \\ Privacy-preserving mechanisms \\ Traffic management and control}

\begin{document}
\maketitle
\titlepageext

\section{Introduction}
Big data has been a powerful catalyst for the continuous growth in the transportation sector. Specifically, the rapid advancement of sensing and vehicular technologies (e.g., roadside sensors, connected vehicles, etc.) and the widespread adoption of mobility services have driven an explosive surge in data availability. This surge in data enables stakeholders, including Municipal Authorities (MAs) and Mobility Providers (MPs) such as ride-hailing companies, to devise effective data-driven analytic and decision-making strategies, including traffic state estimation/prediction~\citep{xu2020ge}, traffic control ~\citep{guo2019urban},  routing ~\citep{huang2020multi}, and fleet management~\citep{GUO2023104244}. These emerging smart solutions have significantly enhanced safety, mobility, and sustainability in transportation systems, thereby fostering the expansion of digital transportation marketplaces. It is projected that the global market size of big data analytics in transportation will grow at a compound annual growth rate of 23.8\% from 2025 to 2030, eventually reaching an expected value of 43 billion USD by 2030~\citep{marketsize}.

The key to unlocking the benefits of big data in transportation marketplaces lies in \emph{data collaboration}. 
This is because transportation data is often fragmented across multiple stakeholders, leading to isolated data silos. 
Specifically, each stakeholder typically possesses data that captures only a limited segment of transportation systems, due to the significant variations in the types of services (e.g., traffic control for MAs and ridesharing services for MPs), the temporal and spatial scale of service coverage (e.g., holidays vs. workdays, public transit networks vs. road networks), and the sensing technologies employed (e.g., loop detectors, mobile sensors, etc.). 
Therefore, the data possessed by various stakeholders are often highly complementary, offering significant mutual value due to the intensive interactions between these stakeholders within the marketplace. This lays a strong foundation for data collaboration as a win-win strategy, with a significant potential to improve social welfare (for MAs) and service quality (for MPs) \citep{smichowski2018determinants}. 
The benefits of data collaboration have been extensively demonstrated in existing literature. For example, studies show that MA can significantly enhance traffic state estimation algorithms by integrating operational data from multiple MPs, thereby leading to more efficient traffic management that can improve social welfare~\citep{zheng2018traffic, WANG2024104743, wanga2024collaborating}. Additionally, collaboration among different MPs has been instrumental in developing popular Mobility-as-a-Service (MaaS) solutions, whereby MPs can coordinate each other's service coverage, thereby attracting more travelers to their platforms and expand their market reach~\citep{pantelidis2024mobility, smichowski2018determinants}. 

However, data collaboration is nontrivial, as data sharing can raise \emph{privacy concerns} for stakeholders who share the data. Such privacy concerns are two-fold. First, the shared data (e.g., trajectory data, survey data, etc.)  can contain sensitive information about individual mobility patterns, which can be used to infer personal profiles or socioeconomic status if accessed by any third parties or adversaries. Personal information has been generally protected by increasingly stringent regulations, e.g., the General Data Protection Regulation (GDPR) of the European Union~\citep{gdpr}, the violations of which can lead to high fines~\citep{uber, Telecompaper, google}. Second, apart from individual-level information, the shared data can also be collectively used to infer operational secrets of the stakeholders, e.g., sensitive algorithmic parameters~\citep{he2020optimal, StravaPrivacy}. The leakage of these parameters can undermine the stakeholders' competitive advantages or even operational safety~\citep{yu2025interaction}. These two types of privacy concerns pose a significant challenge to fostering collaboration in practice. For example, private bus operators refused to share operational data with governments in Hong Kong~\citep{newstoday}, and car manufacturers declined to share electric vehicle testing data with industry peers in Europe~\citep{martens2021data}. 

To address privacy concerns, several advanced privacy-preserving mechanisms have been proposed in recent years to facilitate data collaboration while minimizing the disclosure of sensitive information~\citep{jin2022survey, gati2021differentially, ying2022privacysignal, gao2022privacy, zhang2022privacy, tsao2022trust, tsao2023differentially, TAN2024104453, ZHOU2024104885, WANG2024104743}. 
These mechanisms can be broadly categorized into encryption-based and perturbation-based techniques. Encryption-based techniques, such as secure multi-party computation, enable stakeholders to compute functions collaboratively without explicitly sharing their data~\citep{tsao2022trust, TAN2024104453, ying2022privacysignal}. While 
these techniques can yield high data accuracy, they are often computationally intensive to implement, which may make it challenging for transportation stakeholders to share a large amount of data efficiently. 
In contrast, perturbation-based techniques, such as differential privacy (DP) and information-theoretic approaches, allow data owners to produce perturbed versions of the raw data, reducing privacy risks while preserving its utility for data requesters~\citep{chow2018informed, tsao2023differentially, zhang2022privacy, ZHOU2024104885}.
In other words, perturbation-based techniques offer the potential to strike a ``sweet spot'' between sharing raw data and no sharing, optimizing the trade-off between the benefits of data collaboration and the associated privacy concerns. 
Thanks to their interpretability and computational efficiency, perturbation-based techniques have been widely used in practice, e.g., the US Census Bureau using DP to release the 2020 census data~\citep{USCensus}, protecting individual privacy while enabling statistical analysis.  

Despite the potential of privacy-preserving mechanisms, few works investigate their impact on the data-sharing strategies of transportation stakeholders (i.e., MAs and MPs) in data collaboration. Most existing research focuses on individuals, employing theoretical analyses \citep{fainmesser2023digital} and surveys \citep{berke2023drone, cottrill2015location}, which suggest that privacy-preserving mechanisms can enhance individuals' willingness to share data. However, the insights derived from individual-level studies cannot be directly generalized to transportation stakeholders. 
This is due to the strong interactions among stakeholders, which arise from the overlapping and interdependent nature of their operations.
For example, one MP's sharing data with the MA to improve traffic management can result in system-level improvements (e.g., reduced congestion), which can benefit all MPs operating within the same traffic network.
These strong interactions mean that the data shared by one stakeholder can influence the decision-making of others and, in turn, affect their own operations. 
Very few recent works have considered such dynamics~\citep{BUTLER2021103036,liu2022efficient}, e.g., co-optition whereby the data shared for cooperative purposes can simultaneously be used in competitive contexts, especially when stakeholders operate in overlapping services or markets~\citep{smichowski2018determinants}. 
While several studies employ game-theoretical analyses to investigate data collaboration among transportation stakeholders~\citep{liu2022efficient,chen2023r}, 
these works assume that stakeholders only decide whether and how much data to share, without accounting for the potential perturbations introduced by privacy-preserving mechanisms, as illustrated in Figure~\ref{fig-intro}
The incorporation of perturbations inherently degrades the quality of shared data, giving rise to two fundamental challenges. First, degraded data quality may lead to mismatches between data requesters’ requirements and the quality of the shared data, thereby complicating the underlying game structure. Second, the introduction of random perturbations increases the difficulty of accurately estimating the utilities of both the MA and MPs.
To the best of our knowledge, no existing research has modeled the data collaboration while incorporating privacy-preserving mechanisms, leaving a critical gap in understanding their role in fostering data collaboration. 

\begin{figure}[ht]
    \centering
    \includegraphics[width=1\textwidth]{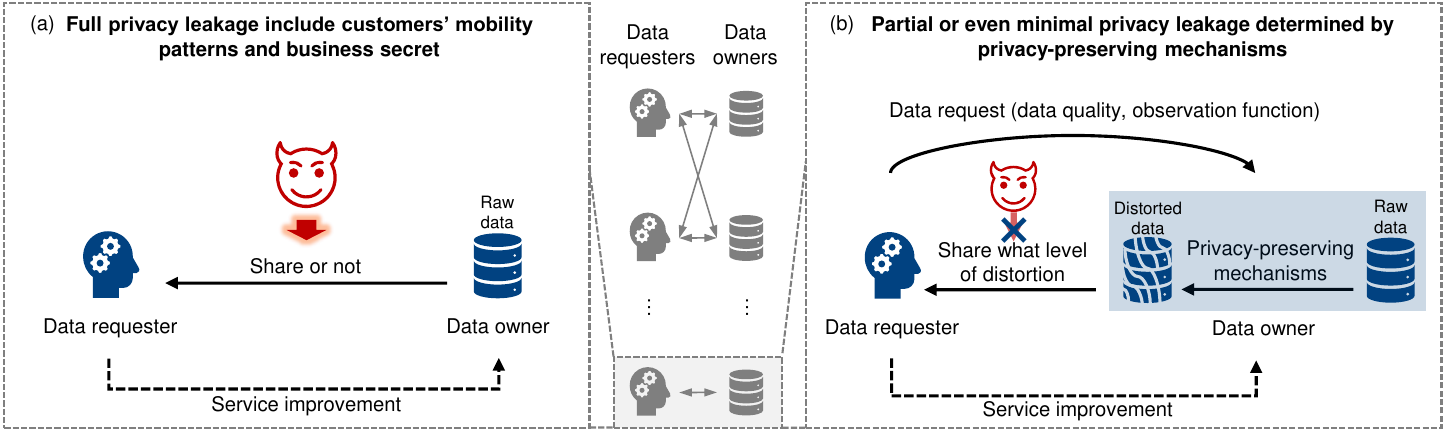}
    \caption{Illustration of data collaboration without (a) and with (b) privacy-preserving mechanisms in transportation systems. Without privacy-preserving mechanisms (a), data sharing is limited to a binary choice and may result in full privacy leakage. With privacy-preserving mechanisms (b), data owners determine the level of data distortion, balancing privacy leakage and data utility.}
    \label{fig-intro}
\end{figure}

\textbf{Statement of Contribution}. To address this critical research gap, we initiate the research to investigate the role of privacy-preserving mechanisms in fostering data collaboration within business-to-government (B2G) digital transportation marketplaces. Our contributions are four-fold. First, we propose a Stackelberg game-theoretic framework to investigate B2G data collaboration among transportation stakeholders, such as MA and MPs, incorporating perturbation-based privacy-preserving mechanisms.
Second, we theoretically analyze the properties of the proposed Stackelberg game (e.g., the existence of equilibrium, conditions for successful data collaboration) under mild practical assumptions.
Third, we instantiate the Stackelberg game-theoretic framework in a practical transportation data collaboration scenario, i.e., collaborative traffic signal optimization, by designing a differentially private traffic demand estimation mechanism that integrates DP with LWR theory.
Fourth, we conduct numerical studies inspired by real-world data to validate the analytical results in typical data-sharing scenarios, providing managerial insights and policy guidance on how and when privacy-preserving mechanisms can promote data collaboration in digital transportation markets. Overall, our work serves as a pioneer in providing a methodological framework and managerial insights to advance the understanding and practice of privacy-preserving data collaboration in digital transportation marketplaces.

This paper is organized as follows. 
Section~\ref{sec: Problem statement} formalizes the data-sharing topologies and introduces the data collaboration problem in a B2G digital marketplace. 
Section~\ref{sec: The model} presents the general model for the data collaboration game and studies the equilibrium through theoretical analysis. 
Section~\ref{sec: Model specification} specifies the model for a transportation-specific setting, collaborative traffic signal optimization.
Section~\ref{sec: Numerical experiments} conducts numerical experiments and provides actionable insights in promoting data collaboration.
Section~\ref{sec: conclusion} concludes the paper.

\section{Problem Statement}\label{sec: Problem statement}
Let us consider a set $\mathcal{K}=\{1,2, \ldots, K\}$ of $K$ stakeholders (each being an MA or MP in our context) who intend to join a data collaboration game in a digital transportation marketplace. 
We use a directed graph $\mathcal{G} = (\mathcal{K}, \mathcal{E})$ to represent the \emph{data-sharing topology} of these stakeholders, where edge $(k,l) \in \mathcal{E}$ represents that stakeholder $k$ \emph{requests} data from stakeholder $l$. We use set $\mathcal{H}_k^+=\{l|(l,k)\in\mathcal{E}\}$ to represent the set of stakeholders that stakeholder $k$ is considering sharing data to and set $\mathcal{H}_k^-=\{l|(k,l)\in\mathcal{E}\}$ to represent the set of stakeholders that stakeholder $k$ is considering requesting data from. 
Based on the data-sharing topology, the role of each stakeholder in data collaboration can be a \emph{data requester} who needs data from other stakeholders and/or a \emph{data owner} who considers sharing data with other stakeholders.
We assume that each stakeholder has at least one of these roles, as otherwise there is no point for them to participate in the data collaboration. 
Mathematically, the set of data requesters can be represented by $\mathcal{K}_R = \{k|\mathcal{H}_k^-\neq\emptyset\}$, and the set of data owners can be represented by $\mathcal{K}_O = \{k|\mathcal{H}_k^+\neq\emptyset\}$. 
We further assume that data sharing is non-transferable, meaning that data received by a requester cannot be further shared with others. This restriction is typically enforced through contractual agreements that explicitly prohibit the transfer of data to third parties. The important variables used in this paper are presented in Table \ref{tab: glossary of terms} of Appendix \ref{sec: variables}.
\begin{remark}[Examples of data-sharing topology]
In real-world transportation data collaboration games with $K$ stakeholders, there are two typical examples of data-sharing topology: 
\begin{enumerate}
    \item[(1)] Star-like topology (see Figure \ref{fig-topology} (a)). In this scenario, an MA requests data from a set of MPs. Mathematically, there exists only one data requester $k$ (i.e., the MA), where $\mathcal{H}_k^- = \mathcal{K}\backslash \{k\}$, and the others (i.e., MPs) are all data owners. For each data owner $l \in\mathcal{K}\backslash \{k\}$, $\mathcal{H}_l^+ = \{k\}$. 
    
    \item[(2)] Fully-connected topology (see Figure \ref{fig-topology} (b)). In this scenario, a set of MPs (e.g., transit operators, ride-hailing companies, etc.) perform data collaboration to share data with each other. Mathematically, each stakeholder is both a data requester and a data owner. For each stakeholder $k$, $\mathcal{H}_k^+ = \mathcal{K}\backslash \{k\}$ and $\mathcal{H}_k^- = \mathcal{K}\backslash \{k\}$.
\end{enumerate}
\end{remark}
\begin{figure}[ht]
    \centering
    \includegraphics[width=0.65\textwidth]{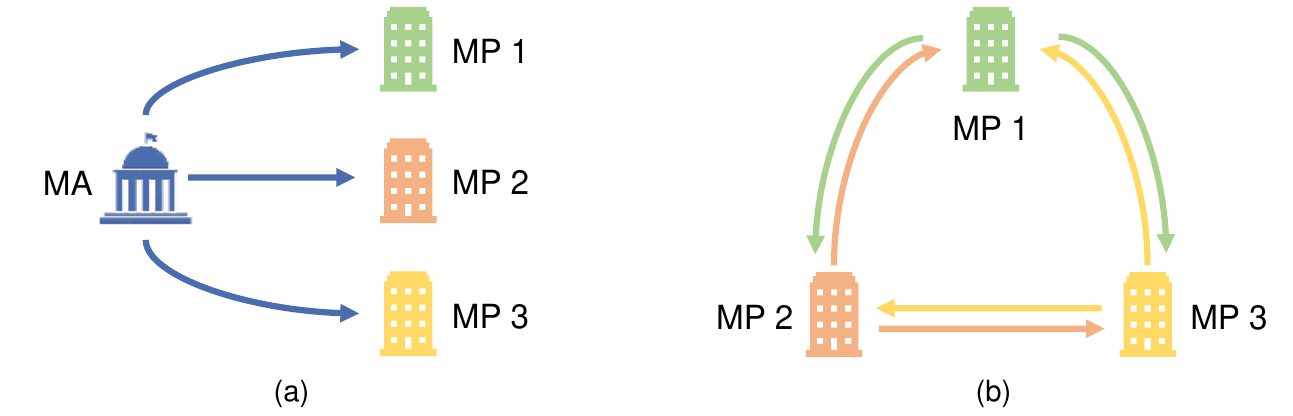}
    \caption{Examples of data-sharing topology. (a) Star-like topology, (b) Fully-connected topology.}
    \label{fig-topology}
\end{figure}

\begin{table}[htbp]
    \centering
        \caption{Examples of star-like data collaboration in digital transportation marketplaces.}
    \begin{tabular}{p{8cm}|p{8cm}}
    \hline\hline
    Types & Examples in research and practice \\\hline
    MA performs efficient traffic management and control with MPs’ traffic observations  & Traffic signal control~\citep{wang2024traffic,wanga2024collaborating}, Traffic state estimation~\citep{WANG2024104743}, 
 \\\hline  
    MA promotes convenient multimodal transport service with MPs’ fleet operations & Mobility as a Service (MaaS) platform~\citep{cottrill2020maas}
 \\\hline  
    MA facilitates innovative deep-tech mobility R\&D with MPs' product validations & Electric vehicle development~\citep{martens2021data}, Autonomous vehicle development~\citep{chen2023r}
    \\\hline\hline
    \end{tabular}
    \label{tab:star-like examples}
\end{table}

In this work, we primarily focus on investigating the star-like data-sharing topology, in which a single MA requests mobility data from multiple MPs. This setting aligns with the structure of a B2G digital marketplace, where private entities supply data to public authorities to support governance, policy implementation, and infrastructure optimization~\citep{chen2023r, wang2024traffic, martens2021data}. Such a topology is particularly relevant in transportation systems and is widely observed in both research and practice, as illustrated by the examples summarized in Table~\ref{tab:star-like examples}.
Specifically, the relevance of the star-like setting is two-fold.
First, many traffic control and management systems are established and operated by MAs, inherently necessitating access to MPs' data to enhance operational efficiency (e.g., MA performs traffic signal control with MPs' vehicle trajectory data \citep{wang2024traffic}).
Second, in many economies, MAs often have priority in formulating and enforcing regulatory policies, giving them an inherent advantage in leading data collaborations across MPs (e.g., MA aggregates MPs' real-time navigation data to support scientific research and facilitate R\&D activities \citep{martens2021data}).
It should be noted that the participation of MPs in the B2G digital marketplace is voluntary and driven by the business interests of MPs. In fact, data sharing can bring tangible benefits to MPs, e.g., service improvement through reduced delay \citep{wang2024traffic} and policy support \citep{martens2021data}.
Hence, in privacy-critical scenarios, MPs can still be incentivized to weigh the potential utility gains from collaboration against their privacy concerns, thereby making strategic decisions regarding data collaboration.

To this end, we consider a star-like data collaboration game that consists of one MA as the data requester, denoted by $\mathcal{K}_R = \{1\}$, and $K-1$ MPs as data owners, denoted by $\mathcal{K}_O = \{2, 3, \cdots, K\}$.

\section{General Model and Analysis}\label{sec: The model}
This section presents a general model and its equilibrium analysis for the data collaboration game between a single MA and multiple MPs in a B2G marketplace, which is typical in transportation systems. 

\subsection{Modeling of the data collaboration game}
We introduce the strategy formulations (see Sections \ref{sssec-str-ma} and \ref{sssec-str-mp}) and utility functions (see Section \ref{sssec-utl}) of the data requester and data owners. Specifically, we model our problem as a Stackelberg game, where a data requester (the leader) selects a strategy to allocate its data requests/expectations, to which the data owners (the followers) respond by deciding whether to engage in data collaboration with the adoption of a privacy-preserving mechanism. 

\subsubsection{Strategy of data requester}\label{sssec-str-ma}
Data requester MA determines and communicates the specifications of the \textit{data request}, a pure strategy, to each relevant data owner $k\in\mathcal{K}_O$. Here, we allow the data requester to customize its request for each data owner, as the data possessed by different data owners may be heterogeneous (e.g., trajectory data under different vehicle penetration rates). The format of a data request is given in Definition~\ref{dfn: request}.

\begin{definition}[Data request] \label{dfn: request}
    A data request sent by data requester MA is structured as a collection $R_1=\{(F_{\sigma_{k}}, \Phi_{k}, d_{k})\}_{k\in \mathcal{K}_O}$, where each tuple $(F_{\sigma_{k}}, \Phi_{k}, d_{k})$ specifies the request sent to data onwer $k\in\mathcal{K}_O$. Here, the function $F_{\sigma_{k}}(\cdot)$ represents an observation function parameterized by $\sigma_{k}$ describing the mapping from the system states $\bm{x}_k$ of stakeholder $k$ to required observation $\bm{y}_{k}$ (i.e., $\bm{y}_{k}=F_{\sigma_{k}}(\bm{x}_k)$).
    The function $\Phi_{k}(\cdot)$ and the scalar $d_{k}$ are used to construct a data quality requirement on the shared data from data owner $k$, i.e., $\Phi_{k}(\bm{y}_{k},\hat{\bm{y}}_{k}) \leq d_{k}$, where $\Phi_{k}(\bm{y}_{k},\hat{\bm{y}}_{k})$ represents the distance between shared data $\hat{\bm{y}}_k$ and true data $\bm{y}_k$  (e.g., $l_2$-norm), and $d_{k}$ represents a threshold. In a data request, the format of the function $\Phi_{k}$ and function $F_{\sigma_{k}}$ as well as parameter $\sigma_{k}$ are pre-determined and known by both MA and MP $k$, and only parameter $\bm{d}=\{d_{k}\}_{k\in \mathcal{K}_O}$ is to be optimized by the data requester.
\end{definition}

We make the following remarks on the data requests. First, the observation function $F_{\sigma_{k}}(\cdot)$ enables data owners to share data with a predefined transformation, which is widely applicable in practice. For example, the observation function may involve data sampling or data aggregation in time and space, both of which are commonly employed to reduce the resolution of shared data. In this context, the aggregation level and sampling frequency can be part of the parameters $\sigma_{k}$. Second, the decision on data quality stems from the possibility of data owners adopting perturbation-based privacy-preserving mechanisms that introduce distortions to the data. The data quality threshold $d_{k}$ represents the data requester's tolerance for the accuracy of the data provided by data owner $k$, as overly inaccurate data may render it unusable for the requester's purposes. Overall, the combination of the observation function and data quality captures the expectations of data requesters for the data.

\subsubsection{Strategies of data owners}\label{sssec-str-mp}
Perturbation-based privacy-preserving mechanisms fundamentally transform data sharing by providing a controllable trade-off between data utility and data privacy. Such a mechanism is defined in Definition~\ref{dfn: pb-ppm}.

\begin{definition}[Perturbation-based privacy-preserving mechanism] \label{dfn: pb-ppm}
    A perturbation-based privacy-preserving mechanism is defined as a randomized function $G_{\epsilon_{k}}(\cdot)$, parameterized by privacy budget $\epsilon_{k}$, that maps true data to its distorted counterpart. Such a function can be realized through noise injection (e.g., differential privacy), data compression (e.g., federated learning), or other perturbation operations. The privacy budget $\epsilon_{k}$ controls the strength of perturbation applied to the true data (i.e., smaller values of $\epsilon_{k}$ correspond to stronger privacy protection and hence stronger perturbation), thereby governing the trade-off between data utility and data privacy.
\end{definition}

Upon receiving the data requests, each data owner $k\in\mathcal{K}_O$ decides the \textit{data-sharing policy} with a privacy-preserving mechanism, also a pure strategy, in response to data request $R_{1}$ from data requester MA. 
The specifications of a data-sharing policy are given in Definition~\ref{dfn: share}.
 
\begin{definition}[Data-sharing policy] \label{dfn: share}
    A data-sharing policy chosen by data owner $k\in \mathcal{K}_O$ to data requester $k\in \mathcal{K}_R$ is structured as a tuple 
    $S_{k}=(a_{k}, G_{\epsilon_{k}})$. Here, $a_{k}\in\{0,1\}$ is a binary variable indicating whether data is shared with the requester MA. If $a_{k}=1$, i.e., data is shared, the function $G_{\epsilon_{k}}(\cdot)$ represents a randomized \emph{perturbation-based privacy-preserving mechanism}. The format of $G_{\epsilon_{k}}$ is predetermined and assumed to be public knowledge, and parameters $(a_{k}, \epsilon_{k})$ are to be optimized by the data owner.
\end{definition}

We make the following remarks on the data-sharing policy. First, we assume that the privacy-preserving mechanism is devised and invested in by data owners for three reasons: i) only data owners have access to these raw data, ii) only data owners know what kind of sensitive information they want to protect, and iii) they seek to enhance market profits with data sharing. Second, we assume that the data owners will strictly follow the requirements as specified in the data request, which can be verified by data requesters using encryption-based methods such as zero-knowledge proof \citep{tsao2022trust}.

\subsubsection{Data player utilities}\label{sssec-utl}
The decision-making processes of all data players in the data collaboration can be formulated into a Stackelberg game. Specifically, based on the data request specifications broadcast by the data requester, the data owners decide their data-sharing policy and compete in a non-cooperative manner. The utility functions for each player are as follows, where we use $k$ and $-k$ to represent the ego data owner and other data owners, respectively.

\noindent \textbf{Follower utility maximization}. After the data requester sends the data request, each data owner $k \in \mathcal{K}_O$ aims to determine its actions $a_k$ and $\epsilon_k$ to maximize a utility function that captures the trade-off between transport-related welfare and privacy cost, which can be written as 
\begin{align}
    \max_{a_k, \epsilon_k}\quad & U_k\left(a_k, \epsilon_k; a_{-k}, \epsilon_{-k}, \bm{d}\right)=W_k\left(\left\{G_{\epsilon_{k'}}(\bm{y}_{k'})\right\}_{k':a_{k'}=1}\right) -  
    \beta_k \epsilon_ka_k,\label{eq: general-UO-obj}\\
    \mathrm{s.t.}\quad& \Phi_k(\epsilon_k) = \bar{\Phi}_k\left(\bm{y}_{k},G_{\epsilon_{k}}(\bm{y}_{k})\right) \leq d_{k}+Z(1-a_{k}),
    \label{eq: general-UO-cons1}\\ 
    &a_k \in \{0,1\},\quad \epsilon_k\in (0,1),\label{eq: general-UO-cons2} 
\end{align}
where the decision variables include a binary variable $a_k$ indicating whether MP $k$ participates in data collaboration and shares data, as well as the privacy budget $\epsilon_k$. 
The objective function \eqref{eq: general-UO-obj} includes (i) the expected transport-related welfare $W_k\left(\cdot\right)$ determined by the MA that depends on all data $\left\{G_{\epsilon_{k'}}(\bm{y}_{k'})\right\}_{k':a_{k'}=1}$ exchanged among the stakeholders, i.e., both participation and (ii) the privacy leakage penalty measured by the parameter $\epsilon_k$ for the privacy-preserving mechanism.
The second term is weighted by $\beta_k$ to reflect the stakeholder's preference on the trade-off between privacy and utility.
Without loss of generality, we assume that $W_k(\emptyset) = 0$, as no data can be used for improving the welfare.
Constraint \eqref{eq: general-UO-cons1} represents that the expected distance, calculated by $\bar{\Phi}_k(\cdot)$, between distorted data $G_{\epsilon_{k}}(\bm{y}_{k})$ and real data $\bm{y}_{k}$ should be no greater than the data quality threshold $d_{k}$ required by data requester $k$, where $Z$ is a sufficiently large number to ensure that this constraint is only active for data requester $k$ with $a_{k}=1$ (i.e., data is shared).

We make the following regularity assumptions regarding the utility functions to characterize the follower’s behavior.
 
\begin{assumption}[Regularity conditions for utility functions] \label{asm:regularity}
We consider regular utility functions that satisfy the following conditions:

\noindent (1) If any data owner $j$ does not participate in the data collaboration, i.e., $a_j=0$, then its privacy budget should not influence any utility function. In other words, $\forall k$, $U_k$ is independent of $\epsilon_j$. This can be written as
    \begin{align}
    U_k(\cdots, \epsilon_j', \epsilon_{j+1},\cdots,a_{j-1},0,\cdots) = U_k(\cdots, \epsilon_j, \epsilon_{j+1},\cdots,a_{j-1},0,\cdots), \forall k,j\in\mathcal{K}_O, \forall \epsilon_j,\epsilon_j'
\end{align} 

\noindent (2) Reporting data with a privacy budget approaching 0 (i.e., strength of perturbation tending to infinity) is equivalent to not reporting data. Specifically, if $a_j=1$ and $\epsilon_j\rightarrow 0$, the utility of any data owner converges to that in the case $a_j=0$, written as
        \begin{align}
   \lim_{\epsilon_j\rightarrow 0^+} U_k(\cdots, \epsilon_j, \epsilon_{j+1},\cdots,a_{j-1},1,\cdots) = U_k(\cdots, \epsilon_j', \epsilon_{j+1},\cdots,a_{j-1},0,\cdots), \forall k,j\in\mathcal{K}_O, \forall \epsilon_j'
\end{align} 
\end{assumption}

We make the following remarks regarding Assumption~\ref{asm:regularity}. The first condition ensures that assigning privacy budgets to non-active data owners does not influence any utility functions, as they are not part of the data collaboration game. This prevents potential confusion in the utility functions. The second condition states that sharing data maximal perturbation (i.e., data that carry no informative content beyond random noise) neither benefits nor harms the utility of any data owner, under the assumption that the MA possesses sufficient capability to identify and analyse data regardless of its quality.

\noindent \textbf{Leader utility maximization}. Given the response $\bm{a}=\{a_{k}\}_{k\in\mathcal{K}_O}$ and $\bm{\epsilon}=\{\epsilon_{k}\}_{k\in\mathcal{K}_O}$ of all followers, the leader, i.e., data requester MA/stakeholder $1$, determines the data request specifications $\bm{d}=\{d_{k}\}_{k\in\mathcal{K}_O}$ by solving the following optimization problem:
\begin{align}
    \max_{\bm{d}\in\mathcal{B}}\quad & U_1\left(\bm{d};\bm{a}, \bm{\epsilon} \right)=W_1\left(\{G_{\epsilon_{k}(\bm{d})}(\bm{y}_k)\}_{k\in\mathcal{K}_O: a_{k} = 1}\right), \label{eq: general-UR-obj}
\end{align}
where the objective function \eqref{eq: general-UR-obj} includes the expected transport-related welfare benefit function $W_1\left(\cdot\right)$ that depends on all data $\left\{G_{\epsilon_{k'}}(\bm{y}_{k'})\right\}_{k':a_{k'}=1}$ exchanged among the stakeholders. Here, without loss of generality, we assume $\mathcal{B}$ is finite, since the distance thresholds can be discretized numerically.

\subsection{Equilibrium Analysis} 
We next analyze the data collaboration game. To this end, we first establish the equilibrium notion in Definition~\ref{dfn: Stackelberg-Nash equilibrium}. 

\begin{definition}[Stackelberg-Nash equilibrium (SNE)]
\label{dfn: Stackelberg-Nash equilibrium}
A strategy profile $(\bm{d}^*, \bm{a}^*, \bm{\epsilon}^*)$ is an SNE of the data collaboration game if and only if:

\noindent (1) $(a_k^*, \epsilon_k^*)$ are the NE of the follower $k$ given the leader strategy $d^*$
\begin{align}
    & (a_k^*, \epsilon_k^*) \in \arg\max_{a_k,\epsilon_k}\; U_k\bigl(a_k,\epsilon_k;\,a_{-k}^*,\epsilon_{-k}^*,\bm{d}^*\bigr),
    \quad \forall\,k\in\mathcal{K}_O,
\end{align}

\noindent (2) $\bm{d}^*$ is the leader’s best response anticipating the followers’ best responses, i.e., 
\begin{align}
    \bm{d}^* \in \arg\max_{\bm{d}\in\mathcal{B}}\; \left\{U_1\!\Bigl(\bm{d};\,\bm{a}^*,\bm{\epsilon}^*\Bigr) \mid (\bm{a}^*,\bm{\epsilon}^*) \in \mathrm{BR}(\bm{d})\right\},
\end{align}
where $\mathrm{BR}(\bm{d})$ represents the followers' best responses to leader decision $\bm{d}$, i.e., 
\begin{align}
    \mathrm{BR}(\bm{d})=\left\{(\bm{a}^*,\bm{\epsilon}^*): (a_k^*,\epsilon_k^*)\in\arg\max_{a_k,\epsilon_k}\; U_k\bigl(a_k,\epsilon_k;\,a_{-k}^*,\epsilon_{-k}^*,\bm{d}\bigr),
    \quad \forall\,k\in\mathcal{K}_O\right\} .
\end{align} 
\end{definition}

Following the equilibrium notion, we next perform equilibrium analysis in Section~\ref{ssec: Equilibrium analysis} and establish conditions for successful data collaboration in Section~\ref{ssec:conditions}. 

\subsubsection{Existence of equilibrium} \label{ssec: Equilibrium analysis}

We begin equilibrium analysis by making several realistic assumptions. Assumption~\ref{asm:distance} specifies a monotonically decreasing relationship between the expected distance and the privacy budget. This is a realistic assumption, as a higher privacy budget implies less noise added to the query statistics, resulting in generated FoQ trajectory points that are better aligned with the true data.

\begin{assumption}[Monotonically decreasing expected distance]\label{asm:distance}
The expected distance between distorted data $\tilde{\bm{y}}_{k} = G_{\epsilon_{k}}(\bm{y}_{k})$ and true data $\bm{y}_{k}$, denoted by $\Phi_k(\epsilon_k) = \bar{\Phi}_k\left(\bm{y}_{k},G_{\epsilon_{k}}(\bm{y}_{k})\right)$, decreases monotonically as $\epsilon_k$ increases. 
\end{assumption}

With Assumption~\ref{asm:distance}, we can rewrite Constraint~\eqref{eq: general-UO-cons1} as 
\begin{align}
    \phi_k \leq \epsilon_k \leq 1,~k\in\mathcal{K}_O,~\label{eq-UO-cons1-1}
\end{align}
where $\phi_k = \Phi^{-1}(d_k)$, which is well defined due to monotonicity. 

Next, we introduce assumptions on the utility function. Let us begin with a simplified formulation of its expression. 
Following Assumption~\ref{asm:regularity}, if a data owner $k$ chooses not to participate ($a_k=0$), the utility $U_k(\bm{\epsilon}, \bm{a})$ is independent of $\epsilon_k$, i.e., $U_k(\epsilon_k, \epsilon_{-k}, \bm{a}) = U_k(0_k, \epsilon_{-k}, \bm{a})$. To simplify the notation, we define $z_k=a_k\epsilon_k$. Then, the variable $z_k$ measures the level of information contributed by data owner $k$, characterizing both participation and the noise level. 
Accordingly, we define a simplified utility function as
\begin{align}
    \bar{U}_k(\bm{z}) = U_k(\bm{\epsilon},\bm{a})
\end{align}
where $z_k=0$ indicates $a_k=0$. Conversely, $z_k>0$, we can recover $a_k=1$ and $\epsilon_k=z_k$.  

Without loss of generality, let us assume $\bar{U}(\cdot)$ is twice differentiable. In many real transportation systems, the marginal value of data decreases with the increase of information provision. Therefore, the marginal benefits of having more participants and more accurate data also decrease as the privacy budget increases. This is summarized in the following assumption. 

\begin{assumption}[Decreasing marginal benefits of data sharing]\label{asm:marginal}
The marginal benefits of data sharing decrease in terms of variable $\bm{z}$, specifically, $\frac{\partial^2\bar{U}(\bm{z})}{\partial z_k\partial z_j} < 0$. 
\end{assumption}
 
Moreover, the utility of a data owner, whether participating in the system or not, tends to increase with more information provision, since a larger data pool generally leads to improved system services.
\begin{assumption}[Monotonicity in others' actions]\label{asm:monotonicity}
The utility of data owner $k$, $U_k$, is non-decreasing with respect to any other data owner's action, i.e., $\frac{\partial U_k}{\partial z_j} \geq 0, \quad \forall j \neq k$.
\end{assumption}

With these assumptions, we can simplify the data owner $k$'s model as follows: 
\begin{align}
    \max_{a_k, z_k}\quad & \bar{U}_k\left(z_k, z_{-k}\right),\label{eq-UO-obj}\\
    \mathrm{s.t.}\quad& \phi_{k}a_k \leq z_k \leq a_k,
    \label{eq-UO-cons1}\\ 
    &a_k \in \{0,1\},z_k\in [0,1]\label{eq-UO-cons3} 
\end{align}
where $\bar{U}_k$ is concave in $z_k$. With any given $a_k$, this optimization problem is a convex optimization problem. 

Due to the composition of a binary action $\bm{a}$ and a continuous action $\bm{z}$, the followers' game is challenging to analyze directly. We therefore adopt a two-stage approach: the upper stage determines $\bm{a}$, and the lower stage computes the NE given $\bm{a}$.
\begin{definition}[Two-stage equilibrium seeking process for the follower model] \label{dfn:two_stage}
    We adopt a two-stage process, whereby the upper stage determines $\bm{a}$, and the lower stage computes the NE given $\bm{a}$. A NE in this setting is defined as a pair $(\bm{a}^*, \bm{z}^*)$, such that (i) for fixed $\bm{a}^*$, no player can unilaterally increase their utility by deviating from $z^*_k$, i.e., $U_k(z_k^*, z_{-k}^*, \bm{a}^*) \ge U_k(z_k, z_{-k}^*, \bm{a}^*)$, and (ii) at the upper stage, no player unilaterally improve its utility by deviating from $\bm{a}^*$, i.e., $V_k(a_k^*, a_{-k}^*)\geq V_k(a_k, a_{-k}^*),\forall a_k\in\{0,1\},\forall k$, where $V_k(\bm{a})$ denotes the utility of player $k$ at the lower-stage equilibrium corresponding to $\bm{a}$. 
\end{definition}

\begin{prop}[Validity of the two-stage process]
Any equilibrium obtained from the two-stage process is an NE of the original followers’ game, and vice versa.
\end{prop}
The result follows immediately from the definition of a NE: at the equilibrium of the two-stage process, no player can improve their utility by unilaterally deviating in either their participation decision or their privacy budget. The converse holds by the same argument.

This two-stage process results in an algorithm for finding the equilibrium of the original follower game. As there are typically very few MPs (e.g., 2 or 3), we can first enumerate all possible upper-level decisions $\bm{a}$, whereby we compute the Nash equilibrium given $\bm{a}$.

 \begin{theorem}[Existence of lower-stage equilibrium] \label{thm:lower_stage}
    Fix $\bm{a}\in\{0,1\}^K$. Then, there exists a unique lower-stage follower NE. 
\end{theorem}

\begin{proof}
    We notice that (i) $z_k\in[\phi_k,1]$ for active players with $a_k=1$ and $z_k=0$ for non-active players with $a_k=0$, which is non-empty, compact, and convex. The utility function is continuous in $\bm{z}$ and strictly concave in own action. Then, by \cite{glicksberg1952further}, there is a unique NE. 
\end{proof}
As mentioned, the utility function $V_k(\bm{a})$ denotes the utility of player $k$ at the lower-stage equilibrium corresponding to $\bm{a}$. 
Then, we analyze the upper-stage game, which is formulated as a game with finite players, each with binary actions. We first show that the utility function is submodular in Theorem~\ref{thm:submodular} (see Appendix~\ref{app:submodular} for proof). 
\begin{theorem}[Submodularity of the upper-stage binary game]\label{thm:submodular}
The upper-stage utility function \(V_k(\bm{a}):=\bar{U}_k\!\big(\bm{z}^*(\bm{a})\big)\), at the lower-stage equilibrium $\bm{z}^*(\bm{a})$ with fixed $\bm{a}$,  has decreasing differences in \((a_i,a_j)\) for every \(i\neq j\); i.e.,
\begin{align}   
\label{eq:submodular}
V_k(1,1_\ell,a_{-k\ell})-V_k(1,0_\ell,a_{-k\ell})
\;\le\;
V_k(0,1_\ell,a_{-k\ell})-V_k(0,0_\ell,a_{-k\ell}),
\quad \forall\,a_{-k\ell}\in\{0,1\}^{K-2}.s
\end{align}

\end{theorem}

Note that a pure NE always exists in two-player submodular games. However, in general, such games may lack a pure NE when there are more than three players. We provide a counterexample in Appendix~\ref{sec: Counterexample}. Therefore, for cases with more than three players, we compute the mixed equilibrium for the upper-stage game.

Once the followers’ game equilibrium is obtained, the leader MA determines the optimal distance threshold $d_i$ via simple enumeration. To this end, we can demonstrate the existence of an SNE for the proposed data collaboration game.
\begin{theorem}[Existence of a SNE for data collaboration]\label{thm:Existence}
    There exists an SNE for the data collaboration game.
\end{theorem}
\begin{proof}
By Theorems~\ref{thm:lower_stage} and \ref{thm:submodular}, for any leader decision $\bm d$, the followers' game admits at least one. The leader can then determine the distance thresholds $\bm{b}\in\mathcal{B}$ via simple enumeration in the finite set $\mathcal{B}$. 
\end{proof}

Finally, we make the following remark regarding the solution process to find an equilibrium of the Stackelberg game. 
\begin{remark}[Solution of SNE]
We make two remarks. First, since the number of MPs is small in a typical transportation data collaboration game, we determine the leader's decision and each follower's upper-stage decision by enumeration. Because each follower's lower-stage problem is strongly convex, the corresponding lower-stage equilibrium is unique. Hence, for any fixed leader's decision and followers' upper-stage decisions, the lower-stage equilibrium can be obtained by solving the associated primal-dual KKT system. 
Second, in more general settings with more data providers, computing an SNE is equivalent to solving a mathematical program with equilibrium constraints (MPEC), which in our setting reduces to a mathematical program with complementarity constraints (MPCC). Since MPCCs typically violate standard constraint qualifications, a mature remedy is to apply a relaxation or regularization of the complementarity constraints and solve a sequence of smooth nonlinear programs (NLPs) that converges to an MPCC solution \citep{scholtes2001convergence,hoheisel2013theoretical}. Each relaxed NLP can then be handled by standard large-scale NLP algorithms such as sequential quadratic programming (SQP) or interior-point methods.  In transportation applications, a common approach is to tackle MPECs via decomposition and alternating updates, leveraging the established variational-inequality (VI) solvers (e.g., diagonalization coupled with SDMSA/MSA) for the lower-level traffic equilibrium, instead of solving the resulting MPCC as a single large-scale NLP \citep{liu2023congested}.
\end{remark}

\subsubsection{Conditions for successful data collaboration} \label{ssec:conditions}
We next analyze the condition for a successful data collaboration (i.e., at least one MP/follower decides to share data), which means that the $\bm a=\bm 0$ is not a strategy profile for the SNE of the data collaboration game. We assume that $\bar U_k(\bm 0)=0$ for no one deciding to share data (i.e., $\bm a=\bm 0$).

We first define the collaboration gain for each MP $k$ in Definition \ref{dfn:Collaboration gain}.
\begin{definition}[Collaboration gain for each MP $k$]\label{dfn:Collaboration gain}
For any fixed $z_{-k}$, the collaboration gain of MP $k$ is defined as
\begin{align}
\Delta_k(\phi_k, z_{-k}) = \max_{z_k\in[\phi_k,1]}\bar{U}_k(z_k,z_{-k}) - \bar{U}_k(0,z_{-k})
\label{eq:Delta}
\end{align}
\end{definition}

The collaboration gain for each MP $k$ has two properties established in Propositions~\ref{prop:phi} and \ref{prop:beta}.

\begin{prop}[Effect of the MA's data quality requirement $d_k$]\label{prop:phi}
For any fixed $z_{-k}$, the collaboration gain $\Delta_k(\phi_k,z_{-k})$ is non-increasing in $\phi_k$.
Equivalently, relaxing the MA's minimum quality requirement (i.e., decreasing $\phi_k$, induced by a larger $d_k$) increases MP $k$'s incentive to share data.
\end{prop}

\begin{proof}
For $\phi_k' \ge \phi_k$, the feasible set satisfies $[\phi_k',1]\subseteq[\phi_k,1]$.
Hence, $\max_{z_k\in[\phi_k',1]}\bar U_k(z_k,z_{-k}) \le \max_{z_k\in[\phi_k,1]}\bar U_k(z_k,z_{-k}),$
Therefore, $\Delta_k(\phi_k',z_{-k}) \le \Delta_k(\phi_k,z_{-k})$.
\end{proof}

\begin{prop}[Effect of privacy cost weight $\beta_k$]
\label{prop:beta}
For any fixed $(\phi_k,z_{-k})$, the collaboration gain $\Delta_k(\phi_k,z_{-k})$ is decreasing in $\beta_k$.
\end{prop}

\begin{proof}
As $\bar U_k(z_k,z_{-k};\beta_k)= \bar U_k(z_k,z_{-k})-\beta_k z_k$, it holds that $\frac{\partial \bar U_k(z_k,z_{-k};\beta_k)}{\partial \beta_k} = -z_k < 0.$
Therefore, $\bar U_k(z_k,z_{-k};\beta_k)$ is decreasing in $\beta_k$. Hence, the collaboration gain $\Delta_k(\phi_k,z_{-k})$ is decreasing in $\beta_k$.
\end{proof}

Then, we can obtain the sufficient condition as follows.
\begin{prop}[Sufficient condition for a successful data collaboration]\label{pps:sufficient_condition}
If there exists an MP $k$ such that
\begin{align}
\Delta_k(\phi_k,\bm 0) > 0,
\label{eq:delta_positive}
\end{align}
then the all-zero profile $\bm a=\bm 0$ cannot be the best response of followers. 
\end{prop}

\begin{proof}
Consider the followers' upper-stage game induced by any leader strategy $\bm d$.
When $\bm{a}=\bm{0}$, MP $k$ obtains payoff $\bar{U}_k(\bm 0)=0$. If $\Delta_k(\phi_k,\bm{0})>0$, there exists a feasible $z_k\in[\phi_k,1]$ such that $\bar{U}_k(z_k,\bm{0}) > \bar{U}_k(\bm{0})=0.$
Hence, MP $k$ has a positive deviation from $a_k=0$ to $a_k=1$.
Therefore, $\bm a=\bm 0$ cannot be the best response of followers.
\end{proof}

We make the following remarks regarding Proposition~\ref{pps:sufficient_condition} to satisfy the condition $\Delta_k(\phi_k,\bm 0)>0$.
First, it is realistic in practice. As MA has no access to traffic observation, a small batch of information from one MP is sufficient to improve $U_k$ significantly (trajectory data with 10\% is enough for traffic state estimation). 
Second, MA can provide a higher $d_k$ when requesting the data, which allows MP $k$ to optimally choose its information level $z_k\in[\phi_k,1]$ to satisfy the condition as described in Proposition \ref{prop:phi}.
Third, MPs may incur additional fixed costs when preparing data (e.g., staff time and operational effort), which can be modeled as a negative constant added to the left-hand side of the participation condition. To compensate for such costs, the MA may offer exclusive benefits to participating MPs (e.g., monetary incentives), which can be represented as a positive constant added to the left-hand side of the condition. This mechanism design is regarded as future work.

\section{Model specification for traffic signal optimization}\label{sec: Model specification}
This section specifies the general model for a collaborative traffic signal optimization problem, whereby a single MA aims to leverage the trajectory data provided by multiple MPs to update its fixed-time traffic signal timing plans.

\subsection{Collaborative signal optimization problem}
We consider a two-way arterial comprising a set of intersections $\mathcal{I}$ with cardinality $I = |\mathcal{I}|$. Each intersection $i \in \mathcal{I}$ operates under a cyclic signal plan with a set of phases $\mathcal{P}_i$, whereby each phase $p \in \mathcal{P}_i$ controls a set of movements at this intersection $\mathcal{M}_{p}$ indexed by $m$. 
Due to the requirement of coordinated signal control, the signal timing plan consists of a common cycle length $C$ for all intersections, green splits $\bm{g} = \{g_{ip}\}_{p\in\mathcal{P}_i, i\in\mathcal{I}}$ of each phase $p$ in intersection $i$, and offsets $\bm{o} = \{o_i\}_{i\in\mathcal{I}}$ for each intersection $i$.

The MA aims to update the outdated signal timing plan $\left(C^{\text{old}}, \bm{g}^{\text{old}}, \bm{o}^{\text{old}}\right)$ to a new optimized plan $\left(C^*, \bm{g}^*, \bm{o}^*\right)$ by applying a signal control policy $\pi$ based on the estimated traffic demand of all movements $\hat{\bm{q}} = \{\hat{q}_m\}_{m \in \mathcal{M}_{p},p\in\mathcal{P}_i,i\in \mathcal{I}}$, i.e., $\left(C^*, \bm{g}^*, \bm{o}^*\right) = \pi \left(\hat{\bm{q}}\right)$.
For simplicity, we use $m \in \mathcal{M}$ to represent all movements in the whole network directly in the following sections.
In this paper, we follow a simple and often-used method that combines (1) the Webster’s formula for calculating signal cycles and green splits with (2) a band-based offset optimization method for calculating offsets. This combined approach has been adopted in many works \citep{little1981maxband, gartner1991multi, zhang2015band} and continues to be used in recent studies and practice. Note that other arterial-level signal controllers can also be integrated into our methodological framework.
  
Both $C$ and $\bm{g}$ can be obtained using Webster's formula
\begin{align}
    &C^* = \max_{i\in \mathcal{I}}\left(\frac{1.5L_i+5}{1-\sum_{p\in \mathcal{P}_i} \max_{m\in \mathcal{M}_{p}}\left(\frac{\hat{q}_{m}}{q_{m}^c}\right)}\right), \label{eq-wb-1}\\
    &g^*_{ip} = \frac{\max_{m\in \mathcal{M}_{p}}\left(\frac{\hat{q}_{m}}{q_{m}^c}\right)}{\sum_{p\in \mathcal{P}_i} \max_{m\in \mathcal{M}_{p}}\left(\frac{\hat{q}_{m}}{q_{m}^c}\right)}C, \quad \forall i\in \mathcal{I}, p \in \mathcal{P}_i, \label{eq-wb-2}
\end{align}
where $L_i$ denotes the total lost time per cycle at intersection $i$, $\max_{m\in \mathcal{M}_{p}}\left(\frac{\hat{q}_{m}}{q_{m}^c}\right)$ denotes the flow ratio of the critical movement of phase $p$. $\hat{q}_{m}$ and $q_{m}^c$ denote the estimated traffic demand and capacity flow of movement $m$, respectively.

The updated offsets $\bm{o}^*$ can be obtained by solving a band-based optimization problem, which aims to maximize the green bandwidth for both outbound and inbound flow and hence reduce vehicle stops. Specifically, this can be formulated as a mixed-integer linear programming problem \citep{little1981maxband, gartner1991multi, zhang2015band}.
We leverage one classic and well-established method, MAXBAND \citep{little1981maxband}, and leverage the adjusted factor for outbound and inbound weight in \citet{gartner1991multi, zhang2015band} to optimize the offsets. In this optimization problem, we have $3I+1$ decision variables, including outbound bandwidth $b$, inbound bandwidth $\bar{b}$, $I$ relative offsets $\bm{\zeta}=\{\zeta_i\}_{i=1}^{I}$ for outbound, $I$ relative offsets $\bar{\bm{\zeta}}=\{\bar{\zeta}_i\}_{i=1}^{I}$ for inbound, and $I-1$ loop integer variable $\bm{M} = \{M_i\}_{i=1}^{I-1}$.
The problem can be written as
\begin{subequations}
\begin{align}
    \max_{b, \bar{b}, \bm{\zeta}, \bar{\bm{\zeta}}, \bm{m}}\quad & \left(\max_{m\in \mathcal{M}^{out}}\frac{\hat{q}_{m}}{q_{m}^c}\right)^\alpha b + \left(\max_{m\in \mathcal{M}^{in}}\frac{\hat{q}_{m}}{q_{m}^c}\right)^\alpha \bar{b}, \label{eq-SC-obj}\\
    \mathrm{s.t.}\quad&
    \zeta_i+b \leq 1-r_i, \quad \forall i \in\{1, \ldots, I\}, \label{eq-SC-cons1} \\&
    \bar{\zeta}_i+\bar{b} \leq 1-\bar{r}_i, \quad \forall i \in\{1, \ldots, I\}, \label{eq-SC-cons2} \\&
    \left(\zeta_i+\bar{\zeta}_i\right)-\left(\zeta_{i+1}+\bar{\zeta}_{i+1}\right)+\left(t_i+\bar{t}_i\right)+\Delta_i-\Delta_{i+1} \notag \\
    & =-\frac{1}{2}\left(r_i+\bar{r}_i\right)+\frac{1}{2}\left(r_{i+1}+\bar{r}_{i+1}\right) +\left(\bar{e}_i+e_{i+1}\right)+M_i \quad \forall i \in\{1, \ldots, I-1\}, \label{eq-SC-cons3} \\&
    M_i\in \mathbb{Z}, \quad \forall i \in\{1 \ldots I-1\}, \label{eq-SC-cons4} \\&
    b, \bar{b}, \zeta_i, \bar{\zeta}_i \geq 0, \quad \forall i \in\{1, \ldots, I\}, \label{eq-SC-cons5}
\end{align} \label{eq:SC}
\end{subequations}
where the objective function (\ref{eq-SC-obj}) includes both the outbound and inbound bandwidth, weighted by the flow ratio of the critical movement $\max_{m\in \mathcal{M}^{out}}\frac{\hat{q}_{m}}{q_{m}^c}$ and $\max_{m\in \mathcal{M}^{in}}\frac{\hat{q}_{m}}{q_{m}^c}$ among all outbound and inbound movements, respectively. $\alpha$ denotes the exponential power with default value $1$. 
Constraints (\ref{eq-SC-cons1}-\ref{eq-SC-cons5}) ensure that the outbound and inbound green bands are feasible at each intersection. $r_i\left(\bar{r}_i\right)$ denotes the outbound (inbound) red time at intersection $i$ in cycles. $\Delta_i$ denotes the time from center of $\bar{r}_i$ to nearest center of $r_i$ in cycles. 
$t_i (\bar{t})$ denotes the travel time from intersection $i$ to intersection $i+1$ outbound (inbound) in cycles.
$e_i\left(\bar{e}_i\right)$ denotes the queue clearance time in cycles, an advance of the outbound (inbound) bandwidth upon leaving intersection $i$.
For more details, readers can refer to \citet{little1981maxband}.

Then, the offsets $o_{i}$ between intersection $i$ and $i+1$ can be computed as follows
\begin{align}
    &o_{i, i+1} + \frac{1}{2}r_{i+1} + \zeta_{i+1} + c_{i+1} = \frac{1}{2}r_{i} + \zeta_{i} + t_{i} \label{eq-SC-cons6}
\end{align}
where the offset for the inbound direction is identical to that of the outbound, as we assume the green splits remain unchanged in both directions. $o_1$ is set to $0$ to calibrate the offset system.

To summarize, MA adopts a control policy $\left(C^*, \bm{g}^*, \bm{o}^*\right) = \pi \left(\hat{\bm{q}}\right)$ that combines Webster's formula with a band-based offset optimization method, as demonstrated in Eqs.~\eqref{eq-wb-1}-\eqref{eq-SC-cons6}.
The resulting average delay of all vehicles traveling along the arterial can be represented by 
\begin{align}
    W_1(C^*, \bm{g}^*, \bm{o}^*, \bm{q}) = W_1 \left(\pi(\hat{\bm{q}}), \bm{q}\right).
\end{align}
where $W_1(\cdot)$ denotes the delay function of all vehicles due to MA's fixed-time signal control policy, which can be calculated by an analytical approach \citep{manual2000highway, yao2019optimization} or a simulation-based approach.
Notice that the resulting average delay can also be represented by a function of the true demand $\bm{q}$ and the estimated demand $\hat{\bm{q}}$ because of the policy $\left(C^*, \bm{g}^*, \bm{o}^*\right) = \pi \left(\hat{\bm{q}}\right)$.

In this signal optimization problem, MA requires the necessary information to estimate key traffic demand parameters $\hat{\bm{q}} = \{\hat{q}_m\}_{m \in \mathcal{M}}$, which are essential in Equations \eqref{eq-wb-1}-\eqref{eq-SC-obj}. We make the following remark on the motivation for the MA to establish data collaboration with MPs.
\begin{remark}[Motivation for MA to establish data collaboration with MPs]\label{rmk:motivation}
The reasons for MA to seek collaboration with MPs are two-fold. First, these key traffic demand parameters are not directly observable by MA due to the lack of sensing devices (e.g., loop detectors). Such a setting is prevalent, as many urban arterial roads still lack reliable coverage of roadside sensors such as loop detectors and traffic cameras due to high installation and maintenance costs, and thus adopt fixed-time signal controllers. This limitation naturally motivates MA to seek/request data from other stakeholders in the transportation network, thereby initiating a data collaboration game. 
Second, MPs can operate a fleet of vehicles within the transportation network, with various market shares and penetration rates. Hence, each of them has access to the trajectory data of its own fleet collected over a recent period of service operations within the arterial. These trajectory data capture fine-grained spatiotemporal movements of vehicles as they traverse signalized intersections, providing rich behavioral insights such as queuing patterns and discharge rates. These key indicators can be leveraged to infer critical traffic demand parameters, whereby the effectiveness of using such data for traffic estimation has been validated in both research and real-world deployments \citep{wang2024traffic}. It is worth noting that trajectory data provide significantly broader and more flexible spatial coverage of traffic observations than loop detectors, which are limited to fixed installation points.
\end{remark}

\subsection{Observation function of data requester (MA)}\label{ssec-str-ma}
Then, we can specify the observation function $F_{\sigma_{k}}(\cdot)$ used by MA in its strategy. Specifically, MA requests trajectory points corresponding to the front of queues (FoQs) to estimate key traffic demand parameters. Let us denote the trajectory data by $\bm{x}_k$ with trajectories indexed by $n\in\mathcal{N}_k = \{1, 2, \cdots, N_k\}$. As defined in Definition~\ref{dfn:obs}, these FoQ trajectory points indicate the time when vehicles join the queue at the intersection due to a red signal. Since the FoQ represents the boundary between stopped and free-flow traffic states, these trajectory points capture critical information about the arrival flow at the intersection, in light of the Lighthill–Whitham–Richards (LWR) theory \citep{lighthill1955kinematic, richards1956shock}. 

\begin{figure}[ht]
    \centering
    \includegraphics[width=0.8\textwidth]{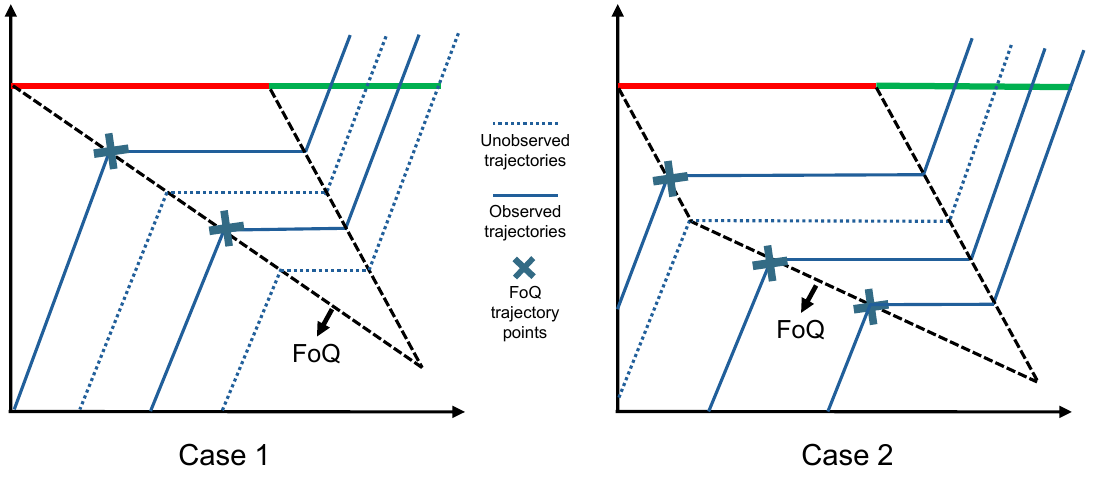}
    \caption{Illustration of the shockwaves on the time-space diagram. The blue dashed lines refer to unobserved vehicle trajectories, while the blue solid lines refer to observed vehicle trajectories. The blue cross-shaped points refer to FoQ trajectory points.}
    \label{fig:foq}
\end{figure}

\begin{definition}[FoQ trajectory points as the observation function]\label{dfn:obs}
For intersection $i$, an FoQ trajectory point is defined as a tuple $(h_{kn}^{m}, t_{kn}^{m})$, 
with $h_{kn}^{m}$ and $t_{kn}^{m}$, respectively, denoting the position and time when vehicle $n$ operated by data owner $k$ joins the queue in movement $m\in\mathcal{M}_{p}, p\in \mathcal{P}_i, i\in\mathcal{I}$. More specifically, the time $t_{kn}^{m}$ is calculated as the elapsed time since the red onset of movement $m$ in the current signal cycle, and $h_{kn}^{m}$ represents the distance from the stopped position to the intersection stop line of movement $m$, measured as a negative value.
Then, we define the observation function $F_{\sigma_k}$ to map vehicle trajectory data $\bm{x}_k$ to corresponding FoQ trajectory points $\bm{y}_k= \big\{\left(h_{kn}^{m},\, t_{kn}^{m}\right)\big\}_{n \in \mathcal{N}_k, m \in \mathcal{M}}$, i.e., $\bm{y}_k = F_{\sigma_k}(\bm{x}_k)$. Notice that only trajectories with a complete stop at the intersection contribute to FoQ trajectory points. 
\end{definition}

With the collected FoQ trajectory points $\bm{y}$, MA will leverage the LWR theory to estimate the arrival flow $\hat{\bm{q}} = B(\bm{y})$, where $B(\cdot)$ represents the LWR theory-based procedure that involves the estimation of the FoQ curve. 
For each movement $m$, we assume the traffic dynamics to be characterized by a triangular fundamental diagram $Q_m(k)=\min\{v_{f,m} k+ w_m(k_{j,m}-k)\}$. 
Based on this model, we can illustrate the shockwaves on the time-space diagram  (see Figure~\ref{fig:foq}). There are two representative cases: 
(1) the arrival flow remains constant within a signal cycle, yielding a linear FoQ curve, and 
(2) the arrival flow starts at saturation flow due to upstream discharge (Seg. a) and then drops to a constant arrival (Seg. b), yielding a piecewise linear FoQ curve with the slope of the first segment being $w_m$. 
In both cases, the arrival flow can be reconstructed from the slope $\psi_m$ of the FoQ curve according to the LWR theory.  

Therefore, we aim to estimate the slope $\psi_m$ of the FoQ curve.  To this end, we consider a simplified scenario where the demand for movement $m$ is constant across all signal cycles\footnote{This assumption can be relaxed by treating the slope $\psi_{m,c}$ of each signal cycle $c$ as a Gaussian-distributed random variable, and formulating the problem using a linear mixed model (LMM). For example, in Case 1, the LMM can be written as $h=\psi_{m,c}t+\xi_m,~\psi_{m,c}\sim N(\psi_{m}, \sigma^2)$. This allows the model to capture practical scenarios with cycle-to-cycle variations in arrival flow. However, this paper does not pursue this generalization, as it is beyond the scope of our primary focus.}. Under this assumption, the FoQ trajectory points from multiple cycles can be aggregated into a single representative cycle, thereby increasing the available data for estimation. We then apply linear regression to fit the FoQ curves in both cases, with the results summarised in Table~\ref{tab:FoQ estimation}. Notice that in Table~\ref{tab:FoQ estimation}, we provide not only the point estimates of the coefficients but also the distributions of the true parameters, derived from standard linear regression theory. 
The distributional estimates are essential for integrating heterogeneous data collected from multiple MPs in a statistically coherent manner. It is worth noting that since MPs often have a long series of data, the sample size $N$ tends to be large. As a result, the Student's $t$-distribution asymptotically approximates the Gaussian distribution. The estimates from different MPs are integrated using a Bayesian framework to obtain the estimates of the true slopes $\psi_m \sim \mathcal{N}(\hat{\psi}_m, \hat{\sigma}_m)$, which will be further used to estimate the arrival flow $\hat{q}_m$ for each movement $m$. The estimated arrival flow $\bm{\hat{q}}$ will be used to derive the signal timing plan with policy $\left(C^*, \bm{g}^*, \bm{o}^*\right) = \pi \left(\hat{\bm{q}}\right)$. 

\begin{table}[htbp]
    \centering
        \caption{FoQ curve estimation. The FoQ curve is modeled as a linear function for Case 1 and a piecewise linear function in Case 2, as shown in Figure~\ref{fig:foq}. In the formulations, $\bm{t}_{k}^{m}=\{t_{k}^{mn}\}_{n\in\mathcal{N}}$ and $\bm{h}_{k}^{m}=\{h_{k}^{mn}\}_{n\in\mathcal{N}}$, respectively, represent the time and space coordinates of FoQ points. The parameters $\psi$ and $\gamma$ represent the slope and intercept to be estimated, while $\Lambda$ collects aggregated quantities directly calculated from FoQ points. Hatted variables, i.e., $\hat{\sigma}$ and $\hat{\psi}$, denote the corresponding statistical estimates obtained from the FoQ points. The variable $\mathcal{T}_n$ represents the student-$t$ distribution with $n$ degrees of freedom.  }
    \begin{tabular}{p{1.8cm}|p{6cm}|p{8.5cm}}
    \hline\hline
    Case & Model & Model estimates \\\hline
    Case 1  &  $h = \psi_{m,k}t + \xi_{m,k},~\xi_{m,k}\sim \mathcal{N}(0,\sigma^2_{m,k})$ & Aggregated quantities from FoQ points: \newline $\displaystyle \Lambda_{mk}^{\text{TH}} = \sum_{n=1}^N t_{kn}^{m}h_{kn}^{m}, ~ \Lambda_{mk}^{\text{T}} = \sum_{n=1}^N (t_{kn}^{m})^2,~\Lambda_{mk}^{\text{H}} = \sum_{n=1}^N (h_{kn}^{m})^2$ \newline
    Statistical estimates of slope $\psi_{m,k}$ and the standard deviation of residuals $\sigma_{m,k}$: \newline 
    $\displaystyle \hat{\psi}_{m,k} = \frac{\Lambda_{mk}^{\text{TH}}}{\Lambda_{mk}^{\text{T}}},~\hat{\sigma}_{m,k}^2 = \frac{1}{N - 1}\left(\Lambda_{mk}^{\text{H}}-\frac{(\Lambda_{mk}^{\text{TH}})^2}{\Lambda_{mk}^{\text{T}}}\right)$\newline
    $\displaystyle\psi_{m,k}\sim \mathcal{T}_{N-1}\left(\hat{\psi}_{m,k}, \frac{\hat{\sigma}_{m,k}^2}{\Lambda_{mk}^{\text{T}}}\right) \approx \mathcal{N}\left(\hat{\psi}_{m,k}, \frac{\hat{\sigma}_{m,k}^2}{\Lambda_{mk}^{\text{T}}}\right)$
 \\\hline  
    Case 2 \newline (Seg. a) &   $h = wt + \gamma_{m,k} + \xi_{m,k},~\xi_{m,k}\sim \mathcal{N}(0,\sigma^2_{m,k})$    
    &     Statistical estimates of intercept $\gamma_{m,k}$ and the standard deviation of residuals $\sigma_{m,k}$: \newline 
    $\displaystyle   \gamma_{m,k} = \frac{1}{N}\sum_{n=1}^N\left(h_{kn}^{m} - w t_{kn}^{m}\right)$ \newline 
    $\displaystyle \hat{\sigma}_{m,k}^2 = \frac{1}{N - 1} \sum_{i=1}^n (h_{kn}^{m} - w t_{kn}^{m} - \gamma_{m,k})^2$ \newline 
    $\displaystyle\gamma_{m,k}\sim \mathcal{T}_{N-1}\left(\hat{\gamma}_{m,k}, \frac{\hat{\sigma}_{m,k}^2}{N}\right) \approx \mathcal{N}\left(\hat{\gamma}_{m,k}, \frac{\hat{\sigma}_{m,k}^2}{N}\right)$
 \\\hline  
    Case 2 \newline (Seg. b) &   $h = \psi_{m,k}t + \gamma_{m,k}' + \xi_{m,k},~\xi_{m,k}\sim \mathcal{N}(0,\sigma^2_{m,k})$    
    &  Statistical estimates of slope $\psi_{m,k}$, intercept $\gamma_{m,k}$, and the standard deviation of residuals $\sigma_{m,k}$: \newline 
    $\displaystyle  \bm{Z}_{m,k} = [\bm{t}_{k}^{m}, \mathbb{I}_{N\times1}] $\newline 
    $[\hat{\psi}_{m,k},\hat{\gamma}_{m,k}]^{\top}=(\bm{Z}_{m,k}^\top\bm{Z}_{m,k})^{-1}\bm{Z}_{m,k}^\top\bm{h}_{k}^{m}$ \newline
    $\displaystyle \hat{\sigma}_{m,k}^2 = \frac{1}{N - 2} \sum_{i=1}^n \|\bm{h}_{k}^{m} - \bm{Z}_{m,k}[\hat{\psi}_{m,k},\hat{\gamma}_{m,k}]^{\top}\|^2$ \newline 
    $\psi_{m,k}\sim\mathcal{T}_{N-2}\left(\hat{\psi}_{m,k}, \hat{\sigma}_{m,k}^2(\bm{Z}_{m,k}^\top\bm{Z}_{m,k})^{-1}\right) \approx \mathcal{N}\left(\hat{\psi}_{m,k}, \hat{\sigma}_{m,k}^2(\bm{Z}_{m,k}^\top\bm{Z}_{m,k})^{-1}\right)$
    \\\hline\hline
    \end{tabular}
    \label{tab:FoQ estimation}
\end{table}

\subsection{Privacy-preserving mechanism of data owners (MPs)}\label{ssec-str-mp}
Then, we specify the privacy-preserving mechanism of MPs in their strategies. Specifically, we introduce the privacy considerations and the privacy-preserving mechanism for each follower MP in Section \ref{ssc: Privacy considerations} and \ref{ssc: Privacy-preserving mechanism}, respectively. The privacy objectives will be formalized and preserved using differential privacy (DP) mechanisms tailored to both trajectory-level and count-level data sensitivity.

\subsubsection{Privacy considerations}\label{ssc: Privacy considerations}
In the traffic signal optimization problem, we use DP to characterize privacy, which is a mathematically rigorous definition of privacy that leads to mechanisms with strong guarantees against privacy leakage from data sharing or analysis. 
Consider a dataset $D = \{\tau_1, \tau_2, \cdots, \tau_N\}$, where each $\tau_i$ represents a data record (e.g., a vehicle trajectory or an FoQ trajectory point). Suppose an external party wishes to query this dataset to extract some statistical insights with a query function denoted by $f(\cdot)$. Typical queries include the total number of data points, the mean of the data points, etc. 
The key idea behind DP is to carefully perturb these statistical outputs such that adjacent datasets (i.e., slightly different datasets) yield statistically indistinguishable outputs, thereby mitigating the risk of inference attacks. 
DP has been widely applied to privacy-preserving data generation, where synthetic data is produced based on perturbed key statistics to ensure privacy while preserving utility.
 
In this work, each data owner $k\in\mathcal{K}_O$ is interested in two types of privacy notions: (1) \textit{trajectory participation privacy}, indicating whether the existence of a particular vehicle's trajectory in the shared dataset can be inferred by adversaries, and (2)  \textit{fleet size or distribution privacy}, indicating whether the aggregate characteristics of the data owner's vehicle fleet, such as the number of active vehicles during certain periods or their spatial distribution, can be inferred by adversaries from the shared data. 
It is worth noting that the fleet size or spatial fleet distribution, if leaked, could reveal operational scale, market share, or demand patterns of the MP. 

To apply DP, we first formalize our privacy notions by specifying the corresponding adjacency relations in Definitions~\ref{dfn:traj_ad} and~\ref{dfn:count_ad}.

\begin{definition}[Trajectory-level adjacency] \label{dfn:traj_ad}
Two datasets $D$ and $D'$ are trajectory-level adjacent, denoted by $D \sim_{\text{traj}} D'$, if $D$ is obtained from $D'$ by removing one data entry (i.e., $D \subset D',~|D|=|D'|-1$) or vice versa (i.e., $D' \subset D,~|D'|=|D|-1$). 
\end{definition}
We make two remarks regarding trajectory-level adjacency in Definition~\ref{dfn:traj_ad}. First, such an adjacency definition corresponds to trajectory participation privacy, and any DP mechanism under this notion seeks to prevent adversaries from inferring the inclusion of individual trajectories. Second, with simple calculation, the condition can also be written as $|D \backslash D'| + |D' \backslash D| = 1$. 

\begin{definition}[Count-level adjacent datasets] \label{dfn:count_ad}
Two datasets $D$ and $D'$ are count-level adjacent, denoted by $D \sim_{\text{cnt}} D'$, if $D$ is obtained from $D'$ by removing $b$ data entries (i.e., $D \subset D',~|D|=|D'|-b$) or vice versa (i.e., $D' \subset D,~|D'|=|D|-b$). 
\end{definition}
We make two remarks regarding count-level adjacency in Definition~\ref{dfn:count_ad}.
First, this adjacency notion naturally corresponds to variations in fleet size or spatial fleet distribution, where a fixed number of vehicle trajectories is added or removed.
Second, the adjacency relationship assumes that when the fleet size changes by $b$, all other data entries remain unchanged. This assumption is reasonable for an MP that provides ride-hailing or mapping services, as when a portion of users leave or join the platform within a short time period (e.g., one month), the remaining users tend to stay due to strong user stickiness.

Next, we introduce the notion of DP for both privacy requirements, as defined in Definitions~\ref{dfn:dp_t} and \ref{dfn:dp_c}. In particular, we adopt $(\epsilon, \delta)$-DP, which requires the output of a randomized mechanism to be statistically indistinguishable. Here, the parameter $\epsilon$ is typically considered as a privacy budget, and the privacy protection becomes more stringent as this value becomes lower. 

\begin{definition}[Trajectory-level $(\epsilon_{\text{traj}}, \delta_{\text{traj}})$-DP]
\label{dfn:dp_t}
For $\epsilon_{\text{traj}}, \delta_{\text{traj}}>0$, a randomized mechanism $M: \mathcal{D} \rightarrow \mathcal{X}$ that maps datasets $D\in \mathcal{D}$ into an output space $M\left(D\right) \in \mathcal{X}$ satisfies $(\epsilon_{\text{traj}}, \delta_{\text{traj}})$-differential privacy, if for any $D \sim_{\text{traj}} D'$ and any event $E \subset \mathcal{X}$, we have
\begin{align}
\mathbb{P}\left[M\left(D\right) \in E\right] \leq e^{\epsilon_{\text{traj}}} \mathbb{P}\left[M\left(D'\right) \in E\right] + \delta_{\text{traj}}
\end{align}
\end{definition}

\begin{definition}[Count-level $(\epsilon_b, \delta_b)$-DP]
\label{dfn:dp_c}
For $\epsilon_b, \delta_b>0$, a randomized mechanism $M: \mathcal{D} \rightarrow \mathcal{X}$ that maps datasets $D\in \mathcal{D}$ into an output space $M\left(D\right) \in \mathcal{X}$ satisfies $(\epsilon_b, \delta_b)$-differential privacy, if for any $D \sim_{\text{cnt}} D'$ and any event $E \subset \mathcal{X}$, we have
\begin{align}
\mathbb{P}\left[M\left(D\right) \in E\right] \leq e^{\epsilon_b} \mathbb{P}\left[M\left(D'\right) \in E\right] + \delta_b
\end{align}
\end{definition}

We now relate the two definitions of DP via Theorem~\ref{thm:conversion}.

\begin{theorem}[Conversion from trajectory-level DP to count-level DP] \label{thm:conversion}
If a randomized mechanism $M$ satisfies trajectory-level $(\epsilon,\delta)$-DP, then it also satisfies count-level $\left(b\epsilon, \delta\sum_{b'=0}^b e^{b'\epsilon}\right)$-DP.
\end{theorem}
\begin{proof}
See Appendix~\ref{thm:conversion}. 
\end{proof}

Theorem~\ref{thm:conversion} indicates that any trajectory-level DP mechanism simultaneously satisfies count-level DP. Therefore, we only need to guarantee trajectory-level DP with a tighter privacy budget. The detailed design of privacy-preserving mechanisms will be discussed in the next subsection.

\subsubsection{Privacy-preserving mechanisms}\label{ssc: Privacy-preserving mechanism}

Building on the privacy notions defined in the previous section, we develop a privacy-preserving mechanism that enables the MP to generate and share FoQ trajectory points with the MA. 

\begin{definition}[Privacy-preserving mechanism $G_{\epsilon_k}$]\label{dfn:ppm}
The privacy-preserving mechanism $G_{\epsilon_k}$ for each data owner $k$ is a mapping that transforms the original FoQ trajectory points $\bm{y}_k = \big\{\left(h_{kn}^{m},\, t_{kn}^{m}\right)\big\}_{n \in \mathcal{N}_k, m \in \mathcal{M}}$ into a distorted version $\hat{\bm{y}}_k = \big\{\left(\hat{h}_{kn}^{m},\, \hat{t}_{kn}^{m}\right)\big\}_{n \in \hat{\mathcal{N}}_k, m \in \mathcal{M}}$ to protect the trajectory level privacy and count level privacy, which can be described as 
\begin{align}
    \hat{\bm{y}}_k  = G_{\epsilon_k}(\bm{y}_k) = R\circ L_{\epsilon_k}\circ f(\bm{y}_k) \label{eq:mechanism}
\end{align}
where the symbol $\circ$ represents the composition of operators. $f(\cdot)$ represents the query function that returns the key statistics to preserve data utility when being used by the MA, as illustrated by Table~\ref{tab:FoQ estimation}. Operator $L_{\epsilon_k}$ represents noise injection to perturb these queries for protecting privacy, and $R$ represents the reconstruction of the shared FoQ trajectory points from the perturbed key statistics. 
\end{definition}

The privacy-preserving mechanism employed in this paper is based on the Gaussian mechanism, as stated in Proposition~\ref{prp:mechanism}. This mechanism adds carefully-designed Gaussian noise to query $f(\cdot)$ based on the sensitivity of the queries to adjacency relations. 

\begin{prop}[Gaussian mechanism] \label{prp:mechanism}
Let $f: \mathcal{X}^n \rightarrow \mathbb{R}^d$ be a query function with global $l_2$ sensitivity $\Delta_f = \sup_{D\sim D'} \left\|f(D)-f\left(D^{\prime}\right)\right\|_2$. For a given data set $D \in \mathcal{X}^n$ and any $\epsilon ,\delta\in(0,1)$, the mechanism that releases 
 $f(D) + \mathcal{N}(0,\sigma^2I_{d\times d})$, where $ \sigma =  \Delta_f \sqrt{2 \log(1.25 / \delta)} / \varepsilon$,  satisfies  $(\varepsilon, \delta)$-DP.
\end{prop}
The proof of Proposition~\ref{prp:mechanism} can be found in \cite{dwork2006differential}. The Gaussian mechanism achieves $(\varepsilon, \delta)$-DP by adding Gaussian noise scaled by the sensitivity of the query function. We next show how the Gaussian mechanism can be applied by the MP to generate privacy-preserving FoQ trajectory points, following the procedure introduced in Eq.~\eqref{eq:mechanism}. To avoid being repetitive, the following discussion will focus on Case 1 in Table~\ref{tab:FoQ estimation}, as the mechanisms for the other cases are similar.

\vspace{0.2em}\noindent \textbf{Query function $f(\cdot)$}. The first step involves carefully designing the query function $f(\cdot)$ to preserve key statistics relevant to the data generation task. 
Inspired by the calculation in Table~\ref{tab:FoQ estimation}, we identify four key variables that MP $k$ is interested in preserving, including $\displaystyle \Lambda_{mk}^{\text{T}} = \sum_{n=1}^N (t_{kn}^{m})^2$, $\displaystyle \Lambda_{mk}^{\text{TH}} = \sum_{n=1}^N t_{kn}^{m}h_{kn}^{m}$,  $ \displaystyle\Lambda_{mk}^{\text{H}} = \sum_{n=1}^N (h_{kn}^{m})^2$, and $N$. In other words, we can define $f(D) = [\rho^{\text{T}}\Lambda_{mk}^{\text{T}},\rho^{\text{TH}}\Lambda_{mk}^{\text{TH}},\rho^{\text{H}}\Lambda_{mk}^{\text{H}},\rho^{\text{N}}N]$, where $\bm{\rho} = [\rho^{\text{T}},\rho^{\text{TH}}, \rho^{\text{H}}, \rho^{\text{N}}]$ is publicly known (to MA and MP) weighting parameters to balance their contribution in the sensitivity. 

\noindent\vspace{0.2em} \textbf{Gaussian mechanism $L_k(\cdot)$}. To design the Gaussian mechanism, we first calculate the sensitivity of including an individual data point on the query outcome. Let $t_{\max}^m$ and $h_{\max}^m$ represent the upper bounds for $t_{kn}^m$ and $h_{kn}^m$, respectively. Then, the sensitivity for the each query variable can be calculated as
\begin{align}
    \Delta_{TH} = t_{\max}^mh_{\max}^m,\quad \Delta_{T} = (t_{\max}^m)^2,\quad \Delta_{H} = (h_{\max}^m)^2, \quad \Delta_N = 1
\end{align}
Then, we can calculate the sensitivity $\Delta_f = \sqrt{\Delta_{TH}^2\rho_{TH}^2 + \Delta_{T}^2\rho_{T}^2 + \Delta_{H}^2\rho_{H}^2+ \Delta_{N}^2\rho_{N}^2}$ and apply Proposition~\ref{prp:mechanism} to generate Guassian noise with $\sigma_f = \Delta_f\sqrt{2 \log(1.25 / \delta)} / \varepsilon$. 

Let $\sigma_{l}^2 = \frac{\rho_{l}^2}{\rho_{TH}^2 +\rho_{T}^2 + \rho_{H}^2+ \rho_{N}^2},~l \in \{\text{N}, \text{T}, \text{TH}, \text{H}\}$ be the share of variance on each query component. Then, we implement the Gaussian mechanism on these query components by adding the following noise. 
\begin{align}
\tilde{N}=N+\eta_{\text{N}}, \quad \eta_{\text{N}} \sim \mathcal{N}\left(0, \sigma_{\text{N}}^2\right), \quad \tilde{\Lambda}_{mk}^{\text{T}}=\Lambda_{mk}^{\text{T}}+\eta_{\text{T}}, \quad \eta_{\text{T}} \sim \mathcal{N}\left(0, \sigma_{\text{T}}^2\right) \\
\tilde{\Lambda}_{mk}^{\text{TH}}=\Lambda_{mk}^{\text{TH}}+\eta_{\text{TH}}, \quad \eta_H \sim \mathcal{N}\left(0, \sigma_{\text{TH}}^2\right),  \quad 
\tilde{\Lambda}_{mk}^{\text{H}}=\Lambda_{mk}^{\text{H}}+\eta_{\textrm{H}}, \quad \eta_n \sim \mathcal{N}\left(0, \sigma_{\text{TH}}^2\right)
\end{align}

\noindent \textbf{Reconstruction of FoQ points $R(\cdot)$}. 
With the perturbed statistics, we can then generate the privacy-preserving synthetic FoQ trajectory points with $h=\tilde{\psi}_{m, k}t+\xi_{m,k},~\xi_{m,k}\sim \mathcal{N}(0,\tilde{\sigma}_{m,k}^2)$, where $\tilde{\psi}_{m, k}$ and $\tilde{\sigma}_{m,k}^2$ are written as
\begin{align}
    \tilde{\psi}_{m, k}&=\frac{\tilde{\Lambda}_{m k}^{\text{TH}}}{\tilde{\Lambda}_{m k}^{\mathrm{T}}} \approx \hat{\psi}_{m, k} + \frac{1}{\hat{\Lambda}_{m k}^{\mathrm{T}}}\left(\eta_{TH} - \hat{\psi}_{m, k}\eta_{T}\right) \\ 
    \tilde{\sigma}_{m, k}& = \frac{1}{\tilde{N}-1}\left(\tilde{\Lambda}_{mk}^{\text{H}} - \frac{(\tilde{\Lambda}_{m k}^{\text{TH}})^2}{\tilde{\Lambda}_{m k}^{\mathrm{T}}}\right)
\end{align}
where the approximation is made when the sample size $N$ is large. We can further analyze the distribution of the FoQ slope as 
\begin{align}
     \tilde{\psi}_{m, k} \sim \mathcal{N}\left(\hat{\psi}_{m, k}, \frac{1}{({\hat{\Lambda}_{m k}^{\mathrm{T}}})^2}\left({\sigma}^2_{\text{TH}}+\hat{\psi}_{m, k}^2 \sigma_{\text{T}}^2\right)\right) \label{eq:beta_dist}
\end{align}
Since $\hat{\Lambda}_{m k}^{\mathrm{T}}$ grows quadratically with the sample size $N$, the impact of the DP noise on the accuracy of $\psi$ diminishes as $N$ increases.

To this end, this privacy-preserving mechanism provides both trajectory-level and count-level privacy guarantees for all data owners.

\subsection{Sampling-based method to determining utility functions}\label{sec: Sampling}
To summarize, we can write follower and leader problems for collaborative traffic signal optimization as follows. 

\noindent (i) \emph{Follower utility maximization}. Each follower solves the same optimization problem as \eqref{eq: general-UO-obj} to \eqref{eq: general-UO-cons2}. Notice that in the $(\epsilon,\delta)$-DP mechanism, we assume that $\delta$ is a given small value and not counted as a decision variable. $\epsilon_k$ refers to the privacy budget defined in Definitions~\ref{dfn:dp_t} and \ref{dfn:dp_c}. These new specifications will not change Assumption \ref{asm:regularity}.

\noindent (ii) \emph{Leader utility maximization}. The leader's utility maximization can be written as the following optimization problem:
\begin{align}
    \max_{\bm{d}}\quad & U_1\left(\bm{d};\bm{a}, \bm{\epsilon} \right)=\mathbb{E}_{\bm{\eta}}\left[W_1\left(C^*, \bm{g}^*, \bm{o}^*, \bm{q}, \bm{\eta}\right)\right], \label{eq-UR-obj}\\
    \mathrm{s.t.}\quad&
    \tilde{\bm{Y}}(\bm{\eta})=\{G_{\epsilon_{k}}(\bm{y}_k, \bm{\eta}_k)\}_{k\in\mathcal{K}_O: a_{k} = 1},\label{eq-UR-cons1}\\&
   \tilde{\bm{q}}(\bm{\eta}) = B(\tilde{\bm{Y}}(\bm{\eta}_k)),
\label{eq-UR-cons2}\\&
    C^*(\bm{\eta}), \bm{g}^*(\bm{\eta}), \bm{o}^*(\bm{\eta}) = \pi(\tilde{\bm{q}}(\bm{\eta})), \label{eq-UR-cons3}
\end{align}
where the objective function \eqref{eq-UR-obj} includes the transport-related welfare benefit function $W_1\left(\cdot\right)$ that depends on the optimal solution $(C^*, \bm{g}^*, \bm{o}^*)$, with $\bm{\eta}_k$ representing the added DP noises by data owner $k$. To obtain the optimal solution, MA first estimates the arrival flow of all movements by Constraint~\eqref{eq-UR-cons2} based on the FoQ trajectory points ($\hat{\bm{Y}}$) among all stakeholders, which can be calculated using Constraint~\eqref{eq-UR-cons1}. Then, the optimal solution can be obtained by solving a sub-optimization problem using Constraint~\eqref{eq-UR-cons3}, which is an MILP fixed-time traffic signal control problem defined in optimization problem~\eqref{eq:SC}.

\noindent \textbf{Determination of the utility functions.} It can be challenging in reality to obtain the specific functional forms in the aforementioned Stackelberg game. Specifically, determining the utility functions of MA and MPs poses two challenges. First, MA does not observe the true traffic demand and only receives perturbed (noisy) data from MPs, making it challenging to compute the transport-related welfare. Second, the MPs' data-sharing strategies are governed by probabilistic mechanisms (e.g., differential privacy), introducing uncertainties to the evaluation of expected utility. 

To address these two challenges, we apply the following techniques: 

\noindent (i) \emph{True demand sampling:}
We model the relationship between the true traffic demand and the perturbed observations shared by the MP. Specifically, we construct a probabilistic mapping from the perturbed shared observations to the underlying true demand (see Eq. \eqref{eq:beta_dist}).  
This allows the MA to sample plausible realizations of the true traffic demand conditioned on the received (noisy) data. Based on these sampled realizations, the MA can then estimate the expected transport-related welfare and quantify how it varies with different privacy levels.

\noindent (ii) \emph{Leader's information briefing:} Based on the utility function defined in the previous section, MA can publish a reference utility function, which quantifies the expected transport-related welfare under different data collaboration conditions (i.e., the condition can be described as a true traffic demand pair with a perturbed traffic demand).
This function can be further represented in a tabular format, termed a data-utility map, indicating the potential average delay that the MA would obtain and could help communicate to each MP, given varying combinations of true and perturbed traffic demand.
\begin{align}
    W_0 = f(\hat{q}_m, q_m), \quad i \in \mathcal{I}, p \in \mathcal{P}_i, m \in \mathcal{M}_p
\end{align}
This construction is consistent with the definitions introduced earlier, as the underlying welfare function for both MA and MPs remains unchanged in the studied data collaboration setting. It should be noted that the values of true and estimated traffic demand used in this table are obtained via a grid search rather than through actual gameplay. As such, the resulting function or table serves as a benchmark utility map that each MP can use to assess whether sharing a particular subset of data would lead to meaningful improvements in transport efficiency.

\section{Numerical experiments}\label{sec: Numerical experiments}
We perform a numerical case study based on real-world data in this section, considering the collaboration between a single data requester (MA) and two data owners (MPs).
This section aims to provide numerical insights to validate the assumptions, analytical results, and equilibrium behavior of the proposed game.
Specifically, Section \ref{ssec: Experiment settings} introduces the basic experiment settings, including network layout and traffic demand. Section \ref{ssec: Results analysis} analyzes the results and provides actionable insights in promoting data collaboration.

\subsection{Experiment settings}\label{ssec: Experiment settings}
\noindent\textbf{Collaboration settings}: We consider a B2G digital transportation marketplace with three stakeholders, whereby one data requester (MA) intends to request data from two data owners (i.e., two active MPs, each operating a vehicle fleet with a penetration rate of 20\%) to update the fixed-time signal timing plan within an urban arterial network.

\begin{figure}[ht]
    \centering
    \includegraphics[width=0.88\textwidth]{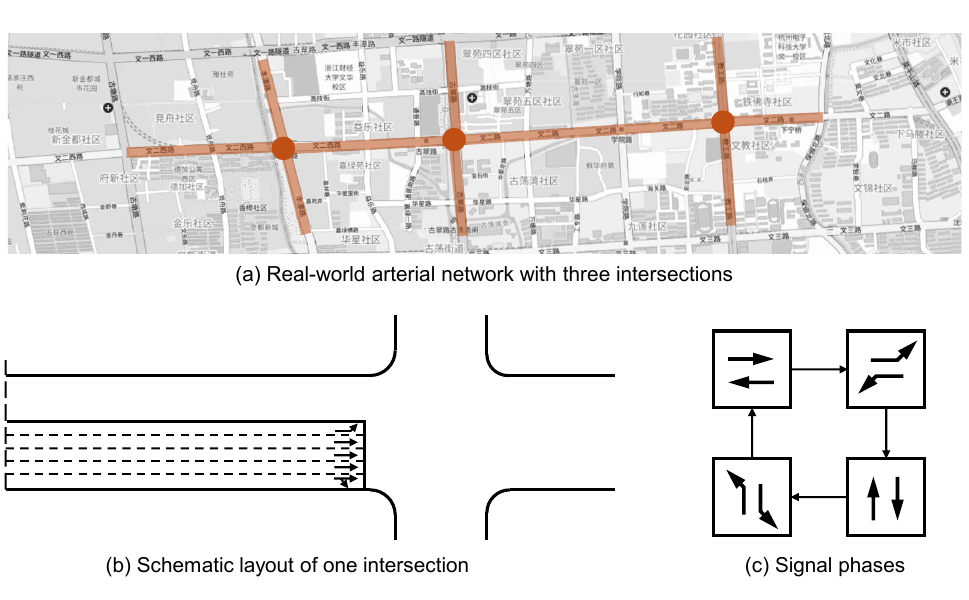}
    \caption{Illustration of the urban arterial network with three intersections.}
    \label{fig-sim}
\end{figure}

\noindent\textbf{Network and demand settings}: We conduct the experiments on a real-world case in Hangzhou, China. The layout, lane configuration, and phase sequence of the urban arterial network with three intersections are shown in Figure \ref{fig-sim}. The lengths of each segment are 922, 989, 1552, and 550 meters, respectively. We simulate the FoQs according to case 1 as introduced in Section \ref{ssec-str-ma}. We set the demand of outbound (i.e., from east to west) through traffic to 3400 veh/h, while 1800 veh/h is set for the inbound (i.e., from west to east) volume. Other left-turning and through traffic along the side street is all configured as 200 veh/h. Under this demand setting, the v/s (volume/saturation flow) ratio for each intersection is 0.93. The traffic demand data is configured to represent realistic peak-hour demand levels, based on traffic demand data collected and processed from traffic camera observations~\citep{wei2019colight, zhu2024coordination, trafficdata}.

\noindent\textbf{Privacy settings}:
The data quality determined by the MA directly corresponds to a set of privacy parameters $\epsilon$ selected by the MPs. Therefore, we assume that the MA can specify the desired data quality based on the scale of $\epsilon$ associated with each data owner. $\beta$ is set to $90$ for each MP.

\begin{figure}[htbp]
     \centering
     \begin{subfigure}[b]{0.49\textwidth}
         \centering
         \includegraphics[width=\textwidth]{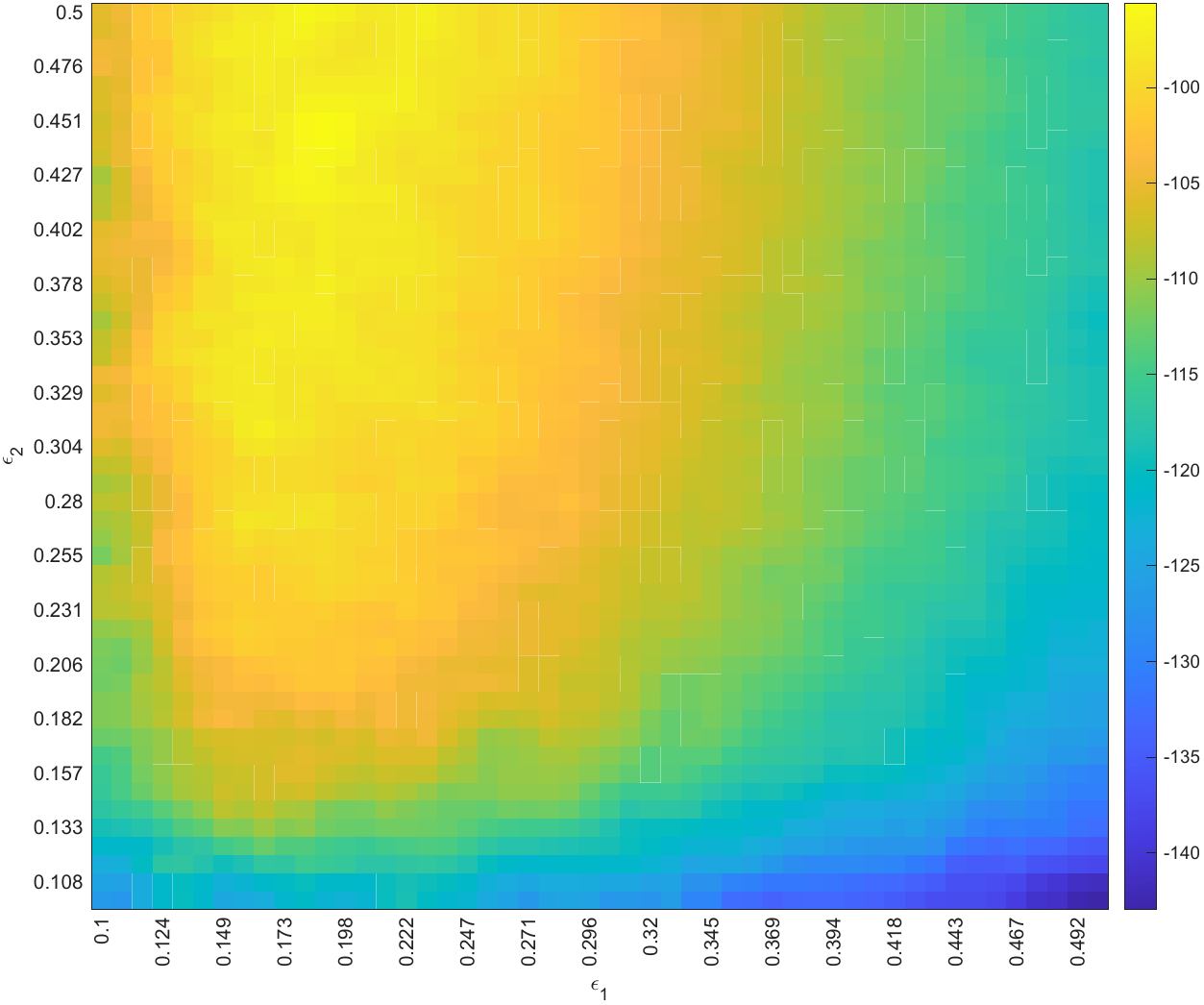}
         \caption{The utility of MP 1.}
         \label{fig:u without a 1}
     \end{subfigure}
     \begin{subfigure}[b]{0.49\textwidth}
         \centering
         \includegraphics[width=\textwidth]{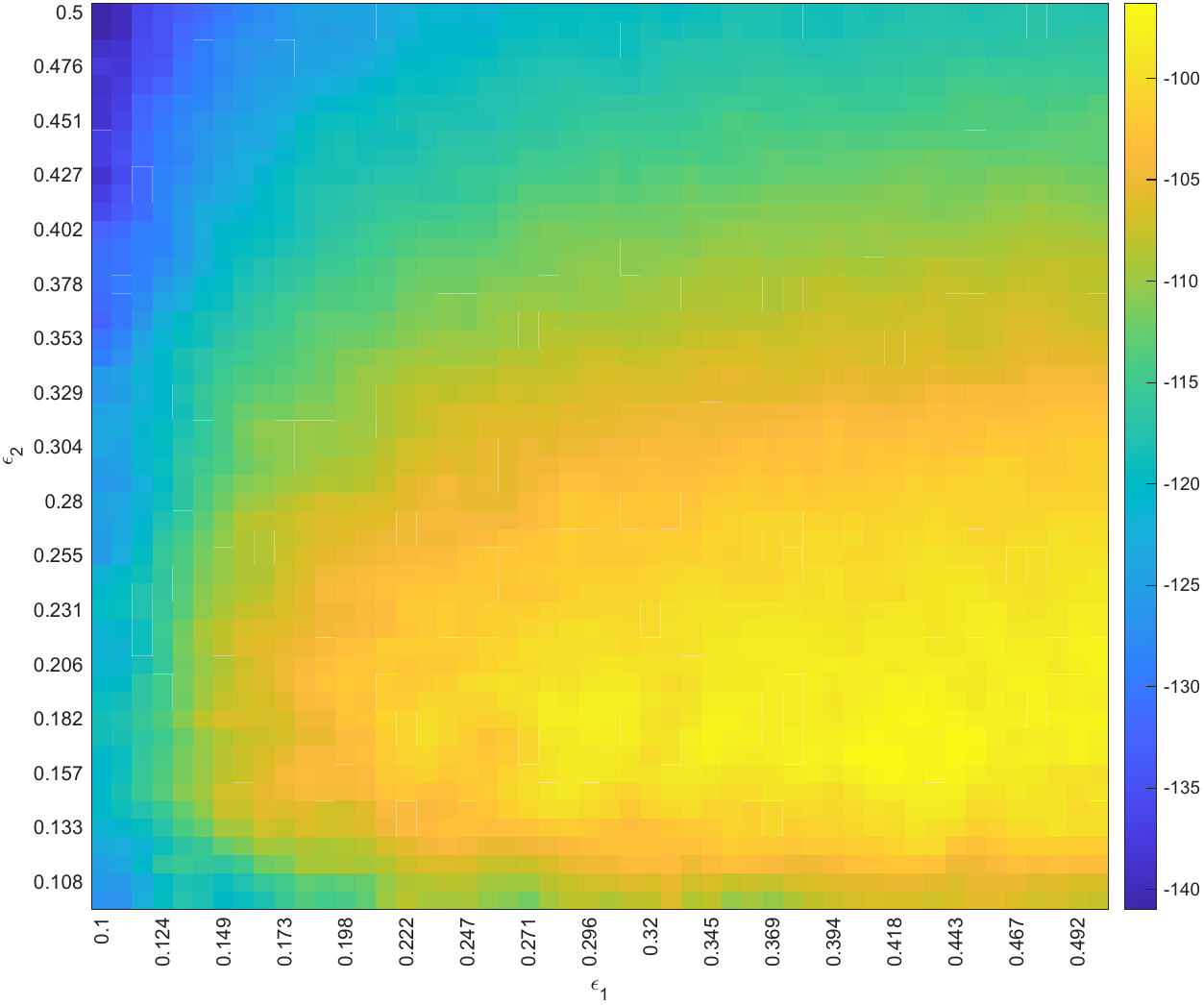}
         \caption{The utility of MP 2.}
         \label{fig:u without a 2}
     \end{subfigure}
     \begin{subfigure}[b]{0.49\textwidth}
         \centering
         \includegraphics[width=\textwidth]{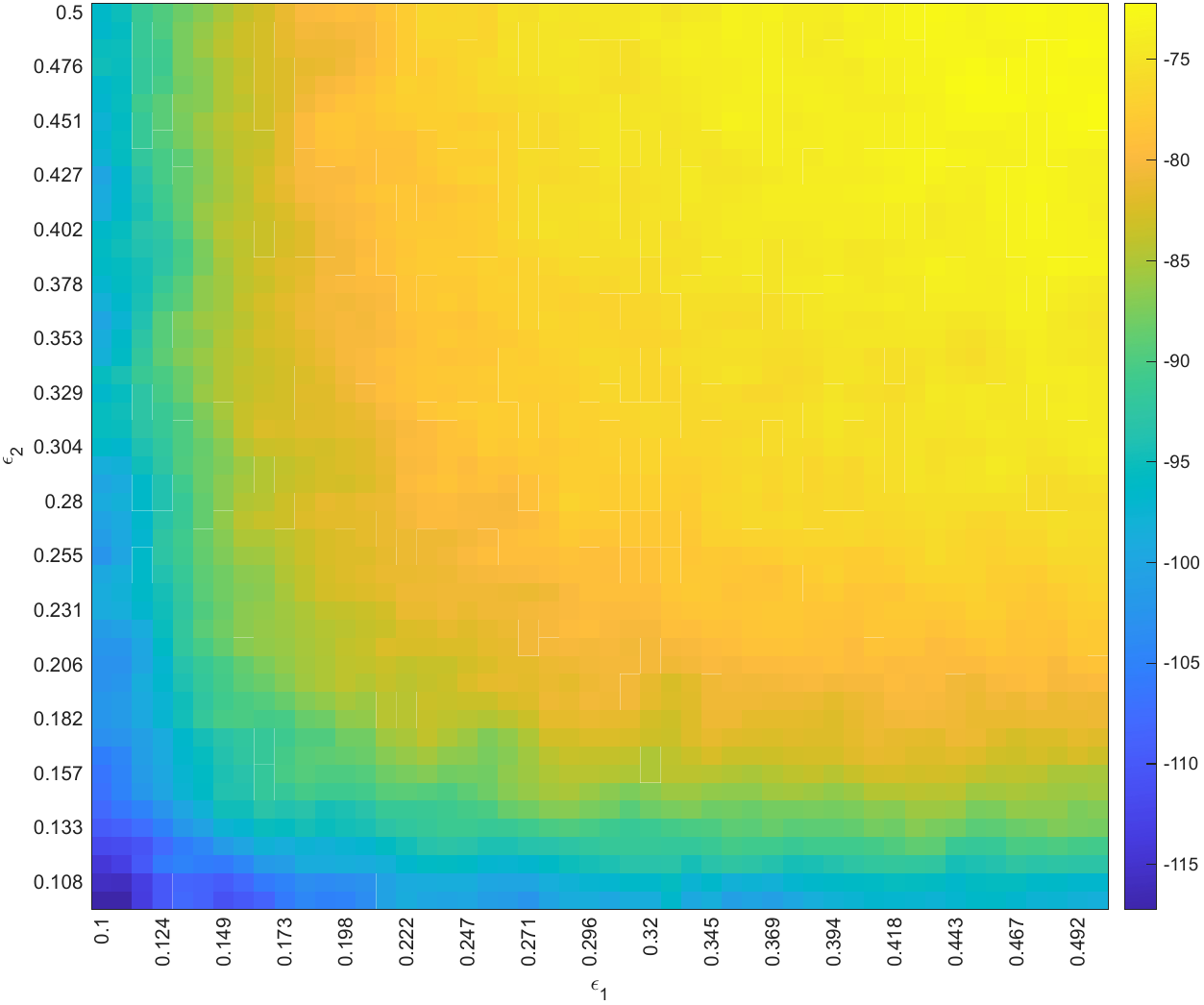}
         \caption{The utility of MA.}
         \label{fig:u without a 3}
     \end{subfigure}
     \caption{The utility of MP 1 and MP 2 with respect to different privacy parameters $\epsilon_1$ and $\epsilon_2$ chosen by MP 1 and MP 2, respectively.}
    \label{fig:u without a}
\end{figure}
\begin{figure}[htbp]
     \centering
     \begin{subfigure}[b]{0.49\textwidth}
         \centering
         \includegraphics[width=\textwidth]{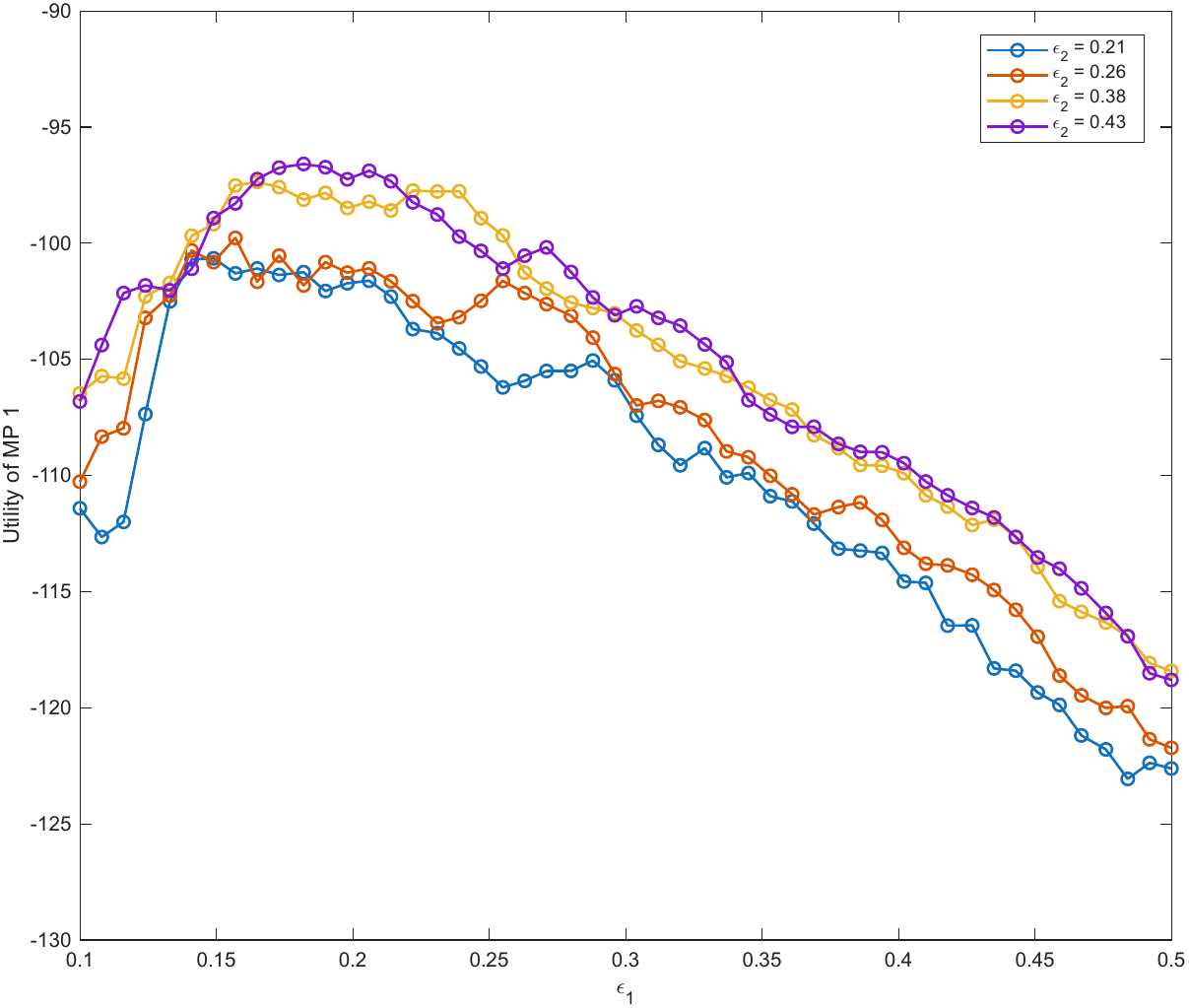}
         \caption{MP 1's utility versus $\epsilon_1$.}
         \label{fig:u-mp1}
     \end{subfigure}
     \begin{subfigure}[b]{0.49\textwidth}
         \centering
         \includegraphics[width=\textwidth]{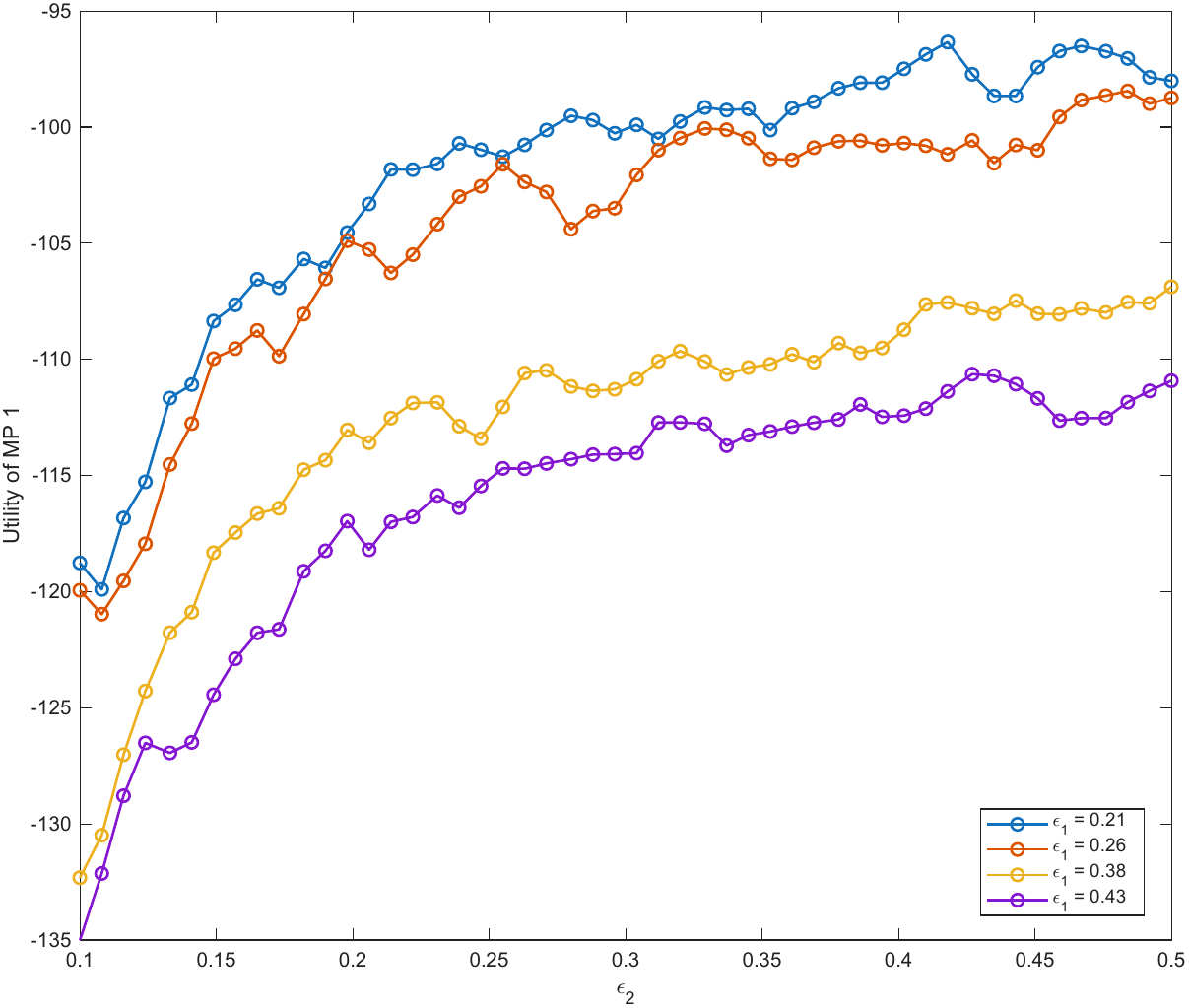}
         \caption{MP 1's utility versus $\epsilon_2$.}
         \label{fig:u-mp2}
     \end{subfigure}
     \caption{The changes in MP 1's utility when one MP's privacy parameter is fixed.}
    \label{fig:u-mp}
\end{figure}

\subsection{Results analysis}\label{ssec: Results analysis}
In this Stackelberg game, the key is to identify the followers’ (MPs’) best responses to the leader’s (MA’s) strategy. We analyze the numerical results from two perspectives. First, we explore the entire strategy space of all MPs to characterize the properties of the utility functions. Second, we calculate all possible best responses of the MPs, which provides insights for determining the leader’s (MA’s) best strategy (i.e., the equilibrium).

Figure \ref{fig:u without a} illustrates the utilities of the MA and MPs as functions of the privacy-preserving parameters $\epsilon_1$ and $\epsilon_2$, chosen by MP 1 and MP 2, respectively. This figure shows how the utilities of all stakeholders, MA and MPs, vary with different combinations of privacy levels selected by the multiple MPs, conditioned on any given pair of data quality parameters $(d_1, d_2)$ chosen by the MA. From the numerical results, we observe the following properties:

\noindent \textbf{More accurate data sharing improves the MA's utility}. As shown in Figure \ref{fig:u without a 3}, the MA’s utility increases monotonically with both $\epsilon_1$ and $\epsilon_2$. The MA’s utility reflects only the expected transport-related welfare, which is obtained from the optimal solution based on the estimated traffic demand. Since a larger $\epsilon$ implies less noise in the shared data, the improvement in estimation accuracy leads directly to higher welfare. This is consistent with the intuition that the MA benefits more from higher-quality data. Correspondingly, the MA heatmap exhibits a warmer color gradient as both $\epsilon_1$ and $\epsilon_2$ increase, indicating higher utility.

\noindent \textbf{MPs exhibit diminishing marginal utility with respect to other MPs’ privacy levels}. As shown in Figures \ref{fig:u without a 1} and \ref{fig:u without a 2}, the utility of each MP increases with the other MP’s privacy parameter $\epsilon$, indicating that more accurate data shared by one participant improves the utility of the others. To further illustrate this effect, Figures \ref{fig:u-mp1} and \ref{fig:u-mp2} plot MP 1’s utility with respect to $\epsilon_1$ and $\epsilon_2$ under various fixed values of the counterpart’s privacy parameter. We observe that, although MP 1’s utility rises as MP 2’s $\epsilon_2$ increases, the rate of improvement decreases progressively. This pattern reflects a diminishing marginal benefit from the increased data accuracy of other participants, which is consistent with Assumption~\ref{asm:marginal}. In other words, while better data from peers always benefits an MP, each additional increment in $\epsilon_j$ yields a smaller improvement in utility compared to the previous increment.

\begin{figure}[htbp]
     \centering
     \begin{subfigure}[b]{0.49\textwidth}
         \centering
         \includegraphics[width=\textwidth]{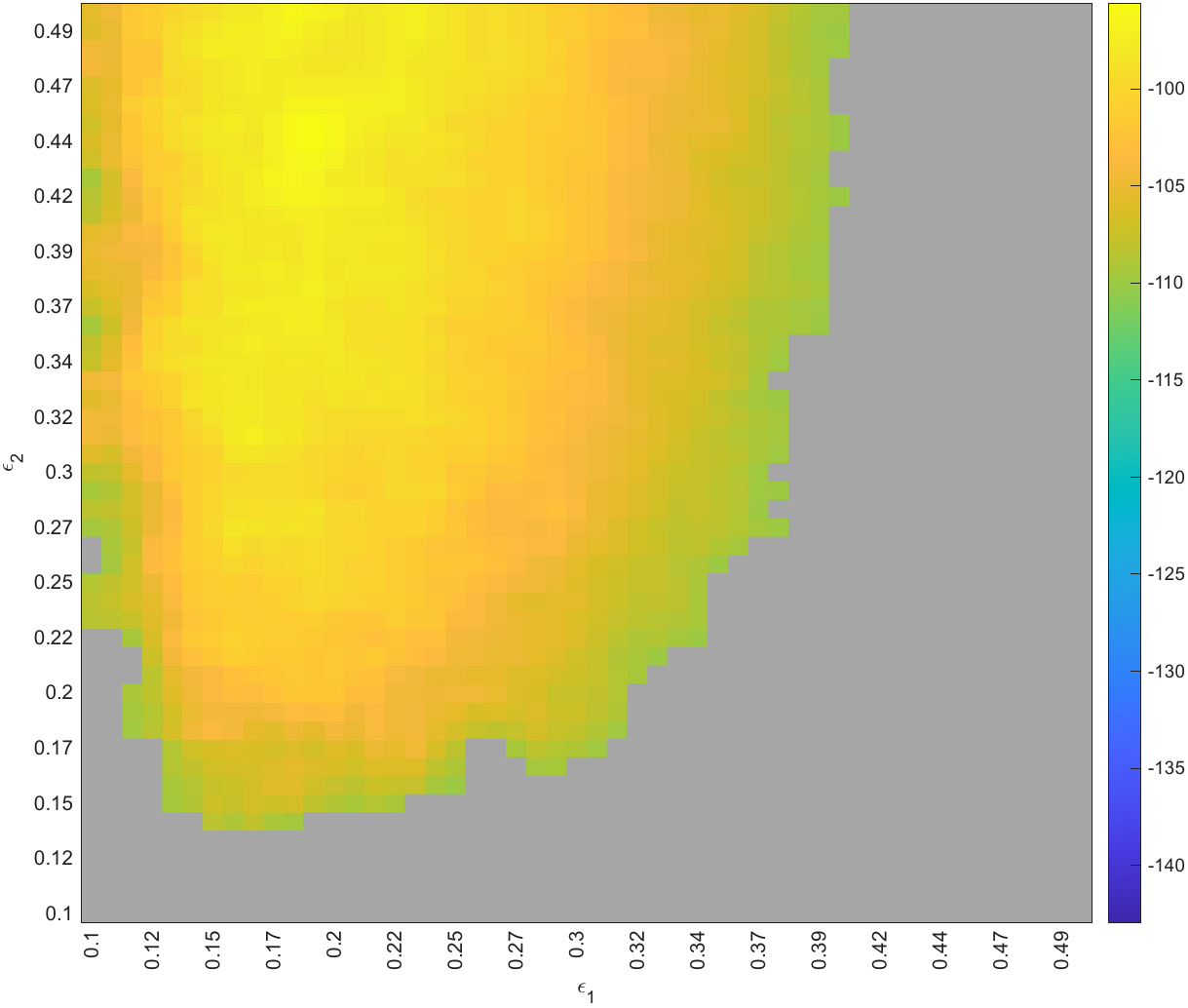}
         \caption{The utility of MP 1.}
         \label{fig:u with a 1}
     \end{subfigure}
     \begin{subfigure}[b]{0.49\textwidth}
         \centering
         \includegraphics[width=\textwidth]{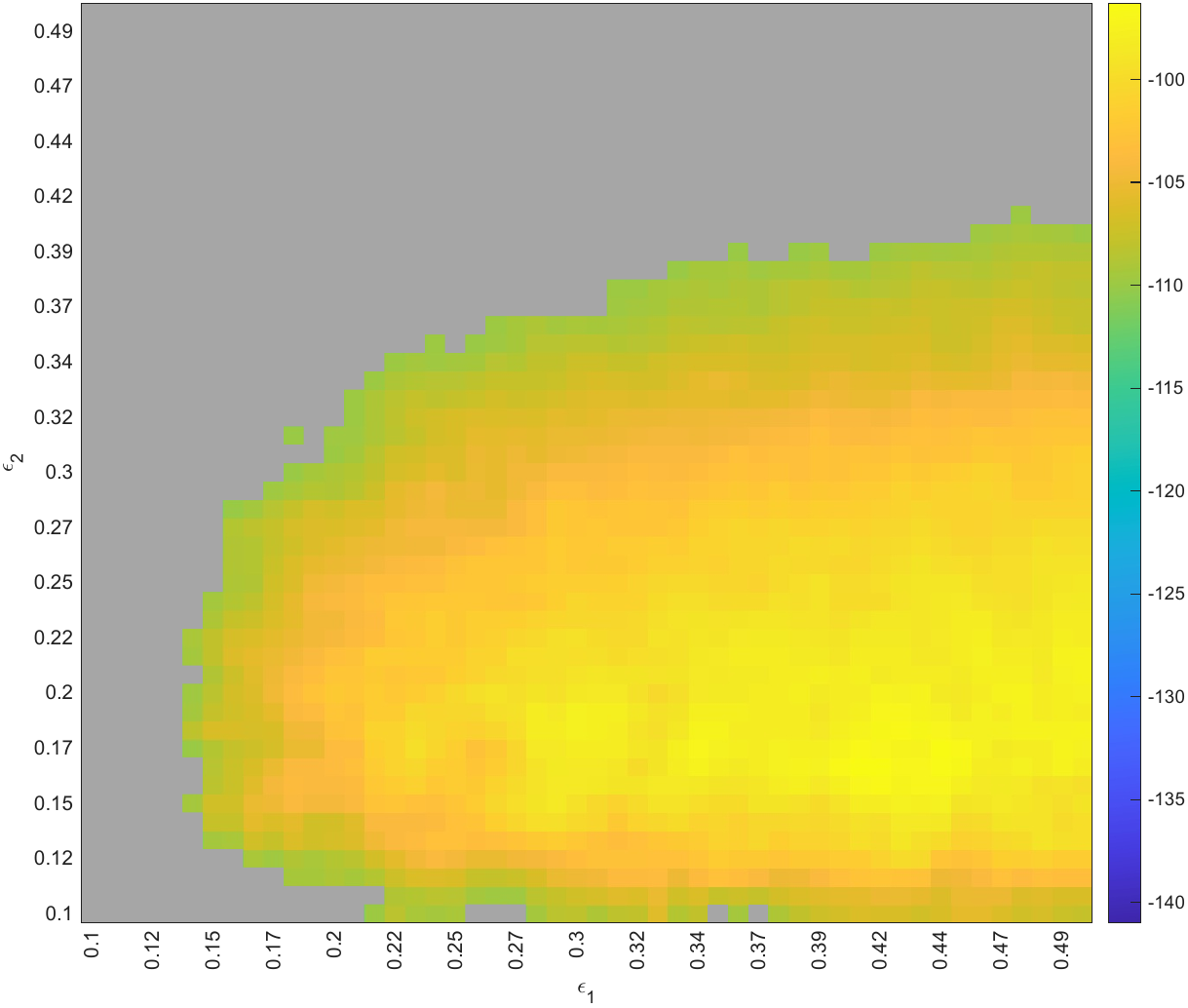}
         \caption{The utility of MP 2.}
         \label{fig:u with a 2}
     \end{subfigure}
     \begin{subfigure}[b]{0.49\textwidth}
         \centering
         \includegraphics[width=\textwidth]{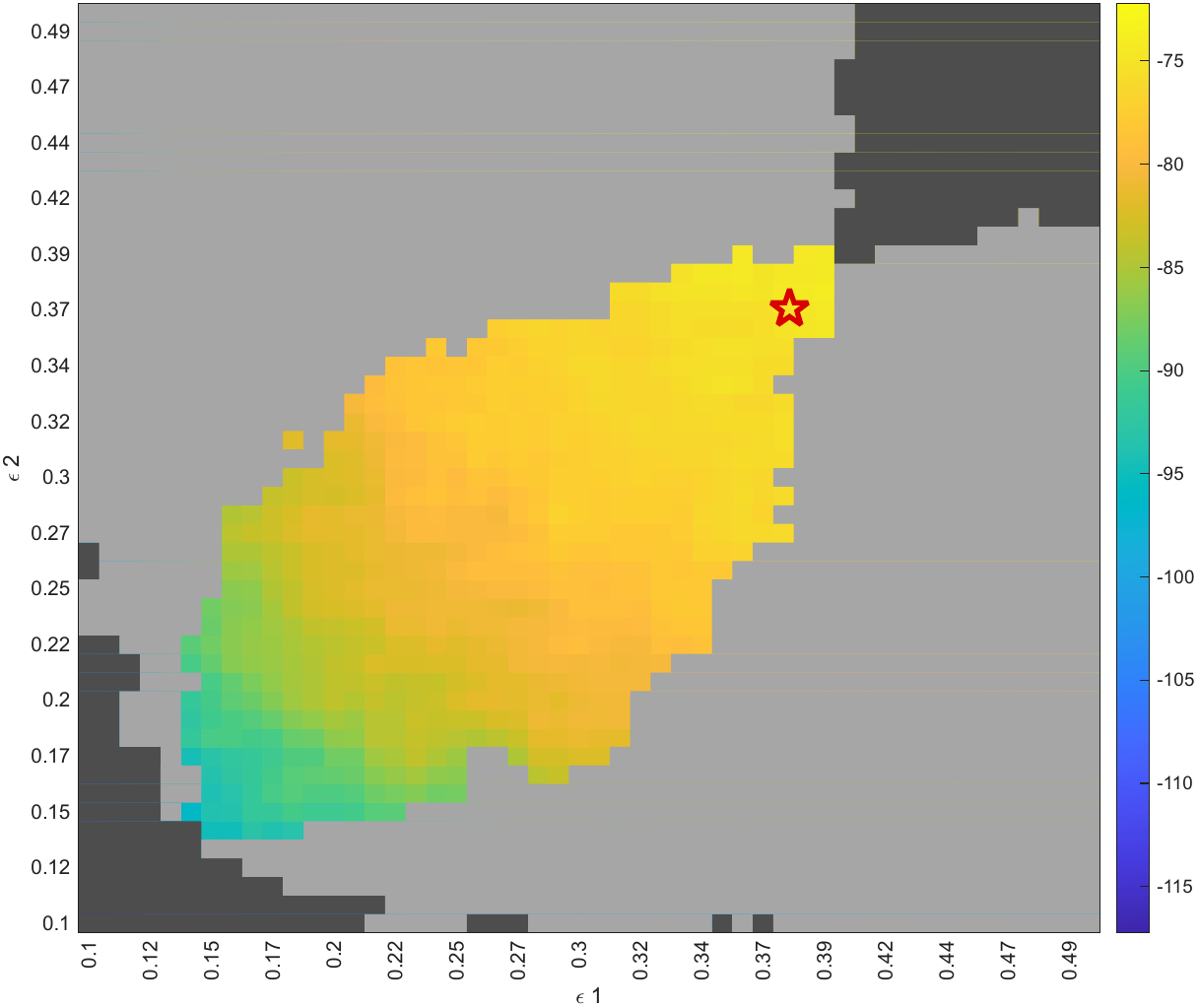}
         \caption{The utility of MA.}
         \label{fig:u with a 3}
     \end{subfigure}
     \caption{The utility of MP 1 and MP 2 with respect to different privacy parameters $\epsilon_1$ and $\epsilon_2$ chosen by MP 1 and MP 2, respectively. Dark gray regions indicate no data sharing. Gray regions indicate that at least one MP shares data. The remaining regions correspond to both MPs sharing data. The red pentagram marks the equilibrium of the data collaboration game.}
    \label{fig:u with a}
\end{figure}

Figure \ref{fig:u with a} presents the utilities of the MA and MPs, where we examine whether an MP has sufficient incentive to share its data. For all subplots in Figure \ref{fig:u with a}, the grey areas mean that there exists at least one MP that is not willing to share the data, where its transport-related benefits obtained from collaboration cannot outweigh the associated privacy costs.
Two key findings are summarized as follows:

\textbf{Excessively high data quality expectations (i.e., larger $d$) can lead to the failure of data collaboration}. 
As shown in Figures~\ref{fig:u with a 1} and~\ref{fig:u with a 2}, a larger $\epsilon$ for each MP results in a lower utility from sharing data, thereby discouraging the MP from participating in the collaboration.
We then plot, in Figure~\ref{fig:u with a 3}, the intersection between the regions where at least one MP is unwilling to share data and the heatmap of the MA’s utility. In the figure, the light gray areas indicate scenarios where one MP chooses not to share, while the dark gray areas indicate scenarios where both MPs choose not to share. Then, we can evaluate the best response of these two MPs based on different $d$ chosen by the MA.
Specifically, when the MA sets a high value for $d$ (i.e., higher data quality requirement), the MP has to choose from a subset of $\epsilon$ with higher values, which incurs a significant privacy cost and results in a utility far below the utility of not sharing data. Consequently, the MP opts not to share data, leaving MA unable to achieve any reward due to the lack of data collaboration.

\textbf{Moderate or low data quality expectations (lower $d$) can promote data collaboration}. 
When MA sets a moderate or lower $d$, the MP's utility for sharing data becomes higher than its utility for not sharing, incentivizing data collaboration (shown in the yellow region in Figure \ref{fig:u with a 3}). Although the data shared by the MP is distorted, the MA receives sufficient data to achieve improved utility. And the privacy cost is relatively acceptable for both MPs. 
The equilibrium in this case lies near the red cross in Figure~\ref{fig:u with a 3}, where the leader (MA) achieves its maximum utility given the MPs’ best responses. In summary, by appropriately balancing data quality requirements with the MPs’ privacy-preserving strategies, the MA can ensure mutually beneficial outcomes for both parties in the game.

\section{Conclusion}\label{sec: conclusion}
In this work, we study the impact of privacy-preserving mechanisms on data collaboration in a B2G digital transportation marketplace. This work is among the first to study such a problem among multiple stakeholders, including MA and MPs. We first introduce a general model that describes our problem as a Stackelberg game, where a data requester (the leader/MA) selects a strategy to allocate its data requests/expectations, to which the data owners (the followers/MPs) respond by deciding whether to engage in data collaboration with the adoption of a perturbation-based privacy-preserving mechanism. This general model captures the complex interplay among privacy protection, data quality, and stakeholder utilities. Then, we specify the model for a transportation-specific application, collaborative traffic signal optimization, whereby the MA leverages the trajectory data provided by multiple MPs to update its fixed-time traffic signal timing plans. We model the privacy-preserving trajectory data sharing by designing a differentially private traffic demand estimation mechanism that integrates DP with LWR theory.

Next, we present equilibrium analysis and numerical experiments based on a real-world case study. Under an idealized setting, we introduce several assumptions and establish the existence of an equilibrium, along with a sufficient condition that guarantees successful data collaboration. Our analysis shows that moderate or low data-quality expectations by the MA (i.e., smaller $d$) can promote data collaboration, whereas excessively high data-quality expectations (i.e., larger $d$) may lead to collaboration failure.
Under a more realistic setting, we develop a solution method and validate the model using a real-world case study. The numerical results verify the assumptions and theoretical insights derived in the idealized analysis. Overall, the outcome of this work provides actionable insights into how and when privacy-preserving mechanisms can facilitate data collaboration in digital transportation marketplaces. The findings are expected to inform policymakers in designing effective strategies to bridge data silos and foster sustainable data collaboration.

There are several future research directions. First, it would be interesting to investigate mechanism designs (e.g., monetary incentives) that enable the MA to encourage broader participation from MPs. 
Second, we are interested in generalizing the proposed Stackelberg framework to more general settings, such as analyzing mixed-strategy equilibria with a large population of followers~\citep{du2014distributed}, and considering incomplete information between the MA and MPs within a Bayesian game framework~\citep{liu2022efficient}.
Third, we would like to extend the current game-theoretic framework to a fully-connected data collaboration topology among multiple MPs. In this topology, the interactions among stakeholders are more complex as the collaboration can influence MPs' operational decisions (e.g., routing choices). 
Fourth, it is worthwhile to extend this framework to other transportation applications beyond traffic control.

\section*{Acknowledgement}
This research was supported by the Singapore Ministry of Education (MOE) under its Academic Research Fund Tier 2 (A-8003064-00-00). This article solely reflects the opinions and conclusions of its authors and not Singapore MOE or any other entity.

\appendix

\section{Proof of Theorem~\ref{thm:conversion}}

Let us first prove a lemma on the conversion from count-level adjacency to trajectory-level adjacency: 

\begin{lemma}[Relation between trajectory-level and count-level adjacency]
\label{lem: sequence between two adj datasets}
For any count-level adjacent datasets $D$ and $D'$, there exists a sequence $D_0\sim_t D_1\sim_t\!\cdots\!\sim_t D_b$ where every two consecutive datasets are trajectory-level adjacent, $D=D_0$ and $D'=D_b$.
\end{lemma}
\begin{proof}
Without loss of generality, let us assume $D'\subset D$ and $|D \backslash D'| = b$ (the other case is symmetric). 
Then, we can build a sequence
\begin{align}
    D_0 = D,\quad D_i = D_{i-1}\backslash\{\tau_i\},\quad i=1,2,\dots,b,
\end{align}
where each $\tau_i$ is chosen arbitrarily from $D_{i-1}\backslash D'$.

With Definition~\ref{dfn:traj_ad}, we have 
\begin{align}
    D_{i-1}\sim_{\text{traj}} D_i, \quad i=1,2,\dots,b,
\end{align}
As $D' \subset D$, we also have $D_b=D\backslash (D\backslash D') = D'$. Hence
\begin{align}
    D = D_0 \sim_{\text{traj}} D_1 \sim_{\text{traj}} \cdots \sim_{\text{traj}} D_b = D'
\end{align}
This completes the proof. 
\end{proof}

We next prove Theorem~\ref{thm:conversion}. 
\begin{proof}[Proof of Theorem~\ref{thm:conversion}]
Let us denote any count-level adjacent datasets $D$ and $D'$ that satisfy the relationship $D\sim_{\text{traj}} D'$ by Definition \ref{dfn:count_ad}. Hence, there exists a chain composed of $b$ trajectory-level adjacent datasets as follows by Lemma \ref{lem: sequence between two adj datasets}.
\begin{align}
w    D = D_0 \sim_{\text{traj}} D_1 \sim_{\text{traj}} \cdots \sim_{\text{traj}} D_b = D' \label{eq:proof1}
\end{align}

Because $M$ is $(\epsilon,\delta)$-DP for every trajectory-level adjacent datasets and any event $E\in \mathcal{X}$, we have
\begin{align}
\mathbb{P}\left[M(D_0) \in E\right]
&\leq e^{\epsilon} \mathbb{P}\left[M(D_1) \in E\right] + \delta \\
&\leq e^{2\epsilon} \mathbb{P}\left[M(D_2) \in E\right] + \delta e^{\epsilon} + \delta \\
&\ \ \vdots \notag \\
&\leq e^{b\epsilon} \mathbb{P}\left[M(D_b) \in E\right]  + \delta \sum_{b'=0}^{b-1}e^{b'\epsilon},
\end{align}

As $D_0 = D$ and $D_b = D'$, mechanism $M$ satisfies count-level $(b\epsilon, \delta \sum_{b'=0}^{b-1}e^{b'\epsilon})$-DP.
\end{proof}

\section{Proof of Theorem~\ref{thm:submodular}}\label{app:submodular}

\begin{proof}
By Assumption~\ref{asm:marginal}, each player's best response $BR_k(z_{-k})$ in the lower-stage game is unique and, by the implicit function theorem, the best response 
\begin{align}
    \frac{\partial BR_k(z_{-k})}{\partial z_\ell} = -\frac{\partial^2 \bar{U}_k}{\partial z_\ell\partial z_k}\Big(\frac{\partial^2 \bar{U}_k}{\partial z_k^2}\Big)^{-1} < 0
\end{align}
Therefore, we have $z_k^*(a_k,1_\ell,a_{-k\ell})\leq z_k^*(a_k,0_\ell,a_{-k\ell})$. It is also trivial to get $z_k^*(1,a_{-k})\geq z_k^*(0,a_{-k})$. 

Then, let us define 
\begin{align}
    V_{\Delta,1} = V_k(1,1_\ell,a_{-k\ell})-V_k(1,0_\ell,a_{-k\ell}) \\
    V_{\Delta,0} = V_k(0,1_\ell,a_{-k\ell})-V_k(0,0_\ell,a_{-k\ell})
\end{align}
which can be seen as the integration of $\frac{\partial U_k}{\partial z_\ell} \geq 0$ along the path of $z_{\ell}$ when $a_l$ switches from $0$ to $1$, in the case $a_k=1$ and $a_k=0$, respectively. By Assumption~\ref{asm:marginal}, $\frac{\partial U_k}{\partial z_\ell}(1,a_{-k}) \leq \frac{\partial U_k}{\partial z_\ell}(0,a_{-k})$. As  range of integral $z_k^*(a_k,1_\ell,a_{-k\ell})\leq z_k^*(a_k,0_\ell,a_{-k\ell})$, we can show that $V_{\Delta,1} \leq V_{\Delta,0}$. 
\end{proof}

\section{Counterexample for the existence of pure Nash equilibrium for the upper-stage follower model}\label{sec: Counterexample}
We provide an example of three players who do not have pure NE for the upper-stage follower model with the following form:
\begin{align}
    U_1(x_1, x_2, x_3) &= \log\!\big(x_1 + x_2 + x_3 + b\big)
\;-\; c_1 \, x_1
\;-\; \gamma_{12} \, x_1 x_2
\;-\; \gamma_{13} \, x_1 x_3, \\[4pt]
U_2(x_1, x_2, x_3) &= \log\!\big(x_1 + x_2 + x_3 + b\big)
\;-\; c_2 \, x_2
\;-\; \gamma_{21} \, x_2 x_1
\;-\; \gamma_{23} \, x_2 x_3, \\[4pt]
U_3(x_1, x_2, x_3) &= \log\!\big(x_1 + x_2 + x_3 + b\big)
\;-\; c_3 \, x_3
\;-\; \gamma_{31} \, x_3 x_1
\;-\; \gamma_{32} \, x_3 x_2.
\end{align}
where $U_1$, $U_2$, and $U_3$ represent the utility functions of each follower, $x_1$, $x_2$ and $x_3$ represent the binary decision variable of each follower. $ b = 1.0099,c1 = 0.161, c2 =0.1545, c3 = 0.1102, \gamma_{12} = 1.0622, \gamma_{13} = 0.0979, \gamma_{21} = 0.0521, \gamma_{23} = 0.5048, \gamma_{31}= 0.9145, \gamma_{32} = 0.2694$.

The utilities for each strategy profile is calculated in Table \ref{tab:utility_table}. This example safeties the submodularity in Theorem \ref{thm:submodular}. There does not exist a pure NE.
\begin{table}[htbp]
\centering
\caption{Utilities of all players under all binary strategy profiles}
\label{tab:utility_table}
\begin{tabular}{ccc|ccc}
\toprule
$x_1$ & $x_2$ & $x_3$ & $U_1$ & $U_2$ & $U_3$ \\
\midrule
0 & 0 & 0 & 0.01 & 0.01 & 0.01 \\
0 & 0 & 1 & 0.70 & 0.70 & 0.59 \\
0 & 1 & 0 & 0.70 & 0.54 & 0.70 \\
0 & 1 & 1 & 1.10 & 0.44 & 0.72 \\
1 & 0 & 0 & 0.54 & 0.70 & 0.70 \\
1 & 0 & 1 & 0.84 & 1.10 & 0.08 \\
1 & 1 & 0 & -0.12 & 0.90 & 1.10 \\
1 & 1 & 1 & 0.07 & 0.68 & 0.09 \\
\bottomrule
\end{tabular}
\end{table}

\section{List of important variables}\label{sec: variables}
\begin{longtable}[htbp]{l@{\hspace{6em}}p{13cm}}
\caption{List of important variables} \label{tab: glossary of terms} \\

\toprule
\multicolumn{2}{l}{Sets} \\
\midrule
$\mathcal{K}$ & Set of stakeholders indexed by $k = 1,\cdots, K$ \\
$\mathcal{H}_k^+$ & Set of stakeholders that stakeholder $k$ is considering sharing data to\\
$\mathcal{H}_k^-$ & Set of stakeholders that stakeholder $k$ is considering requesting data from\\
$\mathcal{K}_R$ & Set of data requesters\\
$\mathcal{K}_O$ & Set of data owners\\
$\mathcal{I}$ & Set of intersections indexed by $i = 1,\cdots, I$ \\
$\mathcal{P}_i$ & Set of signal phases of intersection $i$ indexed by $p_i=1,\cdots,P_i$ \\
$\mathcal{M}_p$ & Set of controlled movements of phase $p$\\
$\mathcal{M}$ & Set of all controlled movements in the arterial\\
$\mathcal{N}_k$ & Set of vehicles observed by data owner $k$ within the arterial network during a recent service period \\

\midrule
\multicolumn{2}{l}{Variables of data requester's strategy (MA)} \\
\midrule
$R_1$ & Collection of a data request sent by MA $1$ \\
$F_{\sigma_{k}}$ & Observation function of MA to stakeholder $k$ parameterized by $\sigma_{k}$ describing the mapping from the system states to required observation\\
$\Phi_k$ & Function that measures the distance between shared data and true data\\
$d_{k}$ & A distance threshold \\
$\bm{d}$ & Vector of distance thresholds sent to all data owners by the requester \\
$\bm{y}_k$ & True data \\
$\hat{\bm{y}}_k$ & Shared data with perturbation \\

\midrule
\multicolumn{2}{l}{Variables of data owners' strategy (MPs)} \\
\midrule
$S_k$ & Collection of a data-sharing policy chosen by data owner $k$ \\
$a_{k}$ & A binary variable indicating whether data is shared \\
$\bm{a}$ & Vector of binary variables of all data owners \\
$G_{\epsilon_{k}}$ & Function that represents a perturbation-based privacy-preserving mechanism \\
$\epsilon_{k}$ & The parameter for privacy-preserving mechanism of data owner $k$ \\
$\bm{\epsilon}$ & Vector of parameters for privacy-preserving mechanism of all data owners \\

\midrule
\multicolumn{2}{l}{Variables of players' utilities} \\
\midrule
$U_k(\cdot)$ & Utility function of data owner $k$ \\
$W_k(\cdot)$ & Service quality improvement of data owner $k$ \\
$\beta_k$ & The weight for privacy cost $k$ \\
$U_1(\cdot)$ & Utility function of data requester $1$ \\
$W_1(\cdot)$ & Transport-related welfare of data requester $k$ \\

\midrule
\multicolumn{2}{l}{Variables of equilibrium analysis} \\
\midrule
$(\bm{d}^*, \bm{a}^*, \bm{\epsilon}^*)$ & A strategy profile that constitutes an equilibrium of the game \\
$\bar{U}_k(z)$ & Simplified representation of $U_k(\bm{\epsilon}, \bm{a})$ \\
$\phi_k$ & Simplified representation of $\Phi^{-1}(d_k)$ \\
$z_k$ & Simplified representation of $a_k\epsilon_k$ \\

\midrule
\multicolumn{2}{l}{Variables of collaborative fixed-time traffic signal optimization} \\
\midrule
$\left(C^{\text{old}}, \bm{g}^{\text{old}}, \bm{o}^{\text{old}}\right)$ & Outdated signal timing plan \\
$\left(C^*, \bm{g}^*, \bm{o}^*\right)$ & New optimized signal timing plan \\
$C$ & Common cycle length of all intersections \\
$g_{ip}$ & Green split of phase $p$ in intersection $i$ \\
$o_i$ & Offset of intersection $i$ to $i+1$ \\
$\zeta_i (\bar{\zeta}_i)$ & Relative offset of intersection $i$ \\
$e_i\left(\bar{e}_i\right)$ & queue clearance time (cycles) \\
$q_{m}^c$ & Capacity flow of movement $m$ \\
$\hat{q}_{m}$ & Estimated traffic flow of movement $m$  \\
$q_{m}$ & True traffic flow of movement $m$ \\
$\hat{\bm{q}}$ & Vector of estimated traffic flow of all movements \\
$\bm{q}$ & Vector of true traffic flow of all movements \\

\midrule
\multicolumn{2}{l}{Variables of the specified observation function and privacy-preserving mechanism} \\
\midrule
$\bm{x}_k$ & System states of data owner $k$ \\
$\bm{y}_k$ & Required observation of stakeholder $k$ \\
$\hat{\bm{y}}_k$ & Perturbed required observation of stakeholder $k$ \\
$\tau_{kn}^m$ & Queuing point collected from trajectory $n$ of data owner $k$ on movement $m$ \\
$h_{kn}^{m}, t_{kn}^{m}$ & Position and time of the queuing points of stakeholder $k$'s $n$th vehicle on movement $m$ on a time-space diagram \\
$\hat{h}_{kn}^{m}, \hat{t}_{kn}^{m}$ & Position and time of the perturbed queuing point of stakeholder $k$'s $n$th trajectory on movement $m$ on a time-space diagram \\
$D, D'$ & Adjacent datasets \\

\bottomrule
\end{longtable}
\bibliographystyle{apalike} 
\bibliography{reference}
\end{document}

%% file: preamble.tex
\usepackage{mathptmx}
\usepackage{geometry}
\usepackage{xcolor}
\usepackage{array}
\usepackage{bbding}
\usepackage{multirow}
\usepackage[fleqn]{amsmath}
\usepackage{amssymb}
\usepackage{amsthm}
\usepackage{graphicx}
\usepackage{natbib}
\usepackage{lscape}
\usepackage[hyphens]{url}
\usepackage[hidelinks]{hyperref}
\usepackage{comment}
\usepackage{bm}
\usepackage[title]{appendix}
\usepackage{longtable}
\usepackage{mathtools}
\usepackage{chngcntr}
\usepackage{authblk}
\usepackage{babel}
\usepackage{dsfont}
\usepackage{subcaption}
\usepackage{fancyhdr}
\usepackage{cleveref}
\usepackage{abstract}

\makeatletter
\allowdisplaybreaks
\date{}

\fancypagestyle{plain}{%
	\fancyhf{} 
	
	\lhead{\footnotesize \textcolor{gray}{arXiv preprint}}
	\chead{\footnotesize \textcolor{gray}{\textit{\shortauthors -- \runningtitle}}}
	\rhead{\footnotesize \textcolor{gray}{\thepage}}
}
\makeatletter
\def\blfootnote{\xdef\@thefnmark{}\@footnotetext}
\makeatother

\makeatletter
\def\titlepageext{
	\begin{center}	
	{\parindent0pt
		\rule{0.9\textwidth}{1pt}
		\begin{minipage}[t]{0.25\textwidth}
			\small {\it Keywords:}\\
			\keyword
		\end{minipage}%
		\hspace{3mm}
		\begin{minipage}[t]{0.6\textwidth}
			\small \abstract
		\end{minipage}%
		
		\rule{0.9\textwidth}{2pt}
	}
	\end{center}

	\blfootnote{* Corresponding author. E-mail address: \href{mailto:\corresemail}{\corresemail}.}
}
\makeatother

\usepackage[mathlines, pagewise]{lineno}
\usepackage{etoolbox} 

\newcommand*\linenomathpatchAMS[1]{%
	\expandafter\pretocmd\csname #1\endcsname {\linenomathAMS}{}{}%
	\expandafter\pretocmd\csname #1*\endcsname{\linenomathAMS}{}{}%
	\expandafter\apptocmd\csname end#1\endcsname {\endlinenomath}{}{}%
	\expandafter\apptocmd\csname end#1*\endcsname{\endlinenomath}{}{}%
}

\expandafter\ifx\linenomath\linenomathWithnumbers
\let\linenomathAMS\linenomathWithnumbers
\patchcmd\linenomathAMS{\advance\postdisplaypenalty\linenopenalty}{}{}{}
\else
\let\linenomathAMS\linenomathNonumbers
\fi

\linenomathpatchAMS{gather}
\linenomathpatchAMS{multline}
\linenomathpatchAMS{align}
\linenomathpatchAMS{alignat}
\linenomathpatchAMS{flalign}

\nolinenumbers

\newtheorem{theorem}{Theorem}
\newtheorem{prop}[theorem]{Proposition}
